%% file: main.tex
\newcommand{\ExtendedVersion}[1]{#1} \newcommand{\PaperVersion}[1]{}

\ExtendedVersion{
	\documentclass[sigconf,nonacm=true]{acmart}
}

\PaperVersion{
	\documentclass[sigconf]{acmart}

	\setcopyright{acmcopyright}
	\copyrightyear{2020}
	\acmYear{2020}
	\acmDOI{10.1145/1122445.1122456}

	\acmConference[Woodstock '18]{Woodstock '18: ACM Symposium on Neural Gaze Detection}{June 03--05, 2018}{Woodstock, NY}
	\acmBooktitle{Woodstock '18: ACM Symposium on Neural Gaze Detection, June 03--05, 2018, Woodstock, NY}
	\acmPrice{15.00}
	\acmISBN{978-1-4503-XXXX-X/18/06}
}

\usepackage{algorithm}
\usepackage{algorithmic} 
\usepackage{amsmath}
\usepackage{array}
\usepackage[USenglish]{babel}
\usepackage[utf8x]{inputenc}
\usepackage{graphicx}
\usepackage{url}  
\usepackage[T1]{fontenc}
\usepackage{listings}
\usepackage[linguistics]{forest}
\usepackage{verbatim}

\input{macros_draft}

\input{macros_symbols}
\input{macros_style}

\hypersetup{bookmarksopen=true}
\begin{document}

\title{FedQPL: A Language for Logical Query Plans over Heterogeneous Federations of RDF Data Sources}
\ExtendedVersion{
	\subtitle{(Extended Version)}
	\subtitlenote{This manuscript is an extended version of a paper in the 22nd International Conference on Information Integration and Web-based Applications \& Services (iiWAS2020). The difference to the conference version is that this extended version includes an appendix with the full proofs of the results in the paper~(see page~\pageref{appendix}).}
}

\author{Sijin Cheng}
\email{sijin.cheng@liu.se}
\affiliation{%
  \institution{Dept.\ of Computer and Information Science (IDA), Link\"oping University}
  \city{Link\"oping}
  \country{Sweden}
}

\author{Olaf Hartig}
\email{olaf.hartig@liu.se}
\orcid{0000-0002-1741-2090}
\affiliation{%
  \institution{Dept.\ of Computer and Information Science (IDA), Link\"oping University}
  \city{Link\"oping}
  \country{Sweden}
}

\begin{abstract}
\input{abstract}
\end{abstract}



\maketitle

\input{introduction}
\input{preliminaries}
\input{datamodel}
\input{language}
\input{srcselection}

\input{equivalences}
\input{concluding}

\begin{acks}  
Sijin Cheng's work was funded
by CUGS~(the National Graduate School in Computer Science, Sweden).
Olaf Hartig's work was funded in equal parts by the Swedish Research Council~(%
	Vetens\-kaps\-r{\aa}det,
project reg.\ no.~2019-05655)
and by the CENIIT program at Lin\-k\"oping University~(project no.~17.05).
\end{acks}

\balance

\bibliographystyle{ACM-Reference-Format}
\bibliography{main}

\ExtendedVersion{
	\input{appendix}
}

\end{document}

%% file: macros_draft.tex
\newcommand{\FinalVersion}[1]{#1}\newcommand{\DraftVersion}[1]{}

\FinalVersion{\usepackage[disable,colorinlistoftodos]{todonotes}}
\DraftVersion{\usepackage[colorinlistoftodos]{todonotes}}

\FinalVersion{\newcommand{\removable}[1]{#1}}
\DraftVersion{\newcommand{\removable}[1]{{\color{gray}#1}}}

\newcommand{\hidden}[1]{} 

%% file: macros_symbols.tex
\newcommand{\definedAs}    
	{\mathrel{\mathop:}=}

\newcommand{\fctDomName}{\mathrm{dom}}
\newcommand{\fctDom}[1]{\fctDomName(#1)}

\newcommand{\symAllURIs}{\mathcal{U}} 
\newcommand{\symAllLiterals}{\mathcal{L}} 
\newcommand{\symAllBNodes}{\mathcal{B}} 
\newcommand{\symAllVariables}{\mathcal{V}} 

\newcommand{\symRDFgraph}{G} 

\newcommand{\fctBNodesName}{\mathrm{bnodes}}
\newcommand{\fctBNodes}[1]{\fctBNodesName(#1)} 
\newcommand{\fctVarsName}{\mathrm{vars}}
\newcommand{\fctVars}[1]{\fctVarsName(#1)} 

\newcommand{\symTP}{tp} 
\newcommand{\symBGP}{B} 
\newcommand{\symPattern}{P} 

\newcommand{\eval}[2]{[\![#1]\!]_{#2}}

\newcommand{\iface}{I} 
\newcommand{\ifaceReqLang}[1]{L_\textsf{#1req}} 
\newcommand{\ifaceReqLangTPF}{\ifaceReqLang{TPF-}}
\newcommand{\ifaceReqLangBrTPF}{\ifaceReqLang{brTPF-}}
\newcommand{\ifaceReqLangSPARQL}{\ifaceReqLang{sparql}}
\newcommand{\ifaceReqLangExp}{\rho} 
\newcommand{\ifaceFct}%
	{\varrho}
\newcommand{\ifaceTuple}{(\ifaceReqLang{},\ifaceFct)}

\newcommand{\member}{\textit{fm}} 
\newcommand{\Fed}{F} 

\newcommand{\xxxAlgebraOperatorName}[1]{\textsf{\itshape #1}}
\newcommand{\request}[2]{\xxxAlgebraOperatorName{req}^{\,#1}_{#2}}
\newcommand{\tpAdd}[3]{\xxxAlgebraOperatorName{tpAdd}^{\,#2}_{#3}(#1)}
\newcommand{\tpAddB}[3]{\xxxAlgebraOperatorName{tpAdd}^{\,#2}_{#3}\bigl(#1\bigr)}
\newcommand{\bgpAdd}[3]{\xxxAlgebraOperatorName{bgpAdd}^{\,#2}_{#3}(#1)}
\newcommand{\bgpAddB}[3]{\xxxAlgebraOperatorName{bgpAdd}^{\,#2}_{#3}\bigl(#1\bigr)}
\newcommand{\union}[2]{\xxxAlgebraOperatorName{union}(#1,#2)}
\newcommand{\unionB}[2]{\xxxAlgebraOperatorName{union}\bigl(#1,#2\bigr)}
\newcommand{\join}[2]{\xxxAlgebraOperatorName{join}(#1,#2)}
\newcommand{\joinB}[2]{\xxxAlgebraOperatorName{join}\bigl(#1,#2\bigr)}
\newcommand{\munion}[1]{\xxxAlgebraOperatorName{mu}\lbrace #1 \rbrace}
\newcommand{\munionB}[1]{\xxxAlgebraOperatorName{mu}\big\lbrace #1 \big\rbrace}
\newcommand{\mjoin}[1]{\xxxAlgebraOperatorName{mj}\lbrace #1 \rbrace}
\newcommand{\mjoinB}[1]{\xxxAlgebraOperatorName{mj}\big\lbrace #1\big\rbrace}

\newcommand{\sols}[1]{\mathsf{sols}(#1)}

\newcommand{\sacost}[1]{\textsf{sa-cost}(#1)}

\newcommand{\myequiv}[1]{\mathrel{\overset{\makebox[0pt]{\mbox{\normalfont\tiny $#1$}}}{\equiv}}}

%% file: macros_style.tex
\newcommand{\rdfterm}[1]{\texttt{#1}}

\newcommand{\definedTerm}[1]{\textbf{#1}}

\theoremstyle{definition}
\newtheorem{myexample}{Example}
\newtheorem{mydefinition}{Definition}

\newtheorem{mytheorem}{Theorem}
\newtheorem{myproposition}{Proposition}
\newtheorem{mylemma}{Lemma}
\newtheorem{mycorollary}{Corollary}

%% file: abstract.tex
Federations of RDF data sources provide great potential when queried for answers and insights that cannot be obtained from one data source alone. A challenge for
	planning the execution of queries over
such a federation
is that the federation may be heterogeneous in terms of the types of data access interfaces provided by the federation members.
	This challenge has not received much attention in the literature.
%
	This paper provides
a solid formal foundation for future approaches that aim to address this challenge. Our main conceptual contribution is a formal language for representing query execution plans; additionally, we identify a fragment of this language that can be used to capture the result of selecting relevant data sources for different parts of a given query. As technical contributions, we show
	that this fragment is more expressive than what is supported by existing source selection approaches, which effectively highlights an inherent limitation of these approaches.
Moreover, we show
	that the source selection problem\hidden{ considering our language} is NP-hard and in $\Sigma_2^\mathrm{P}$,
and we
	provide
a comprehensive set of rewriting rules that can be used
	as a basis for query optimization.

%% file: introduction.tex
\section{Introduction} \label{sec:Intro}

Existing research on querying federations of RDF data sources focuses on federations that are homogeneous in terms of the type of
interface
	via which \removable{each of} the federation members can be accessed.
In particular, the majority of
	work
in this context assumes that all federation members provide the SPARQL endpoint interface~(e.g., \cite{%
acosta2011anapsid%
,charalambidis2015semagrow%
,saleem2018costfed%
,schmidt2011fedbench%
,schwarte2011fedx%
,DBLP:journals/tlsdkcs/VidalCAMP16%
})%
\removable{, whereas
	another line of research focuses
solely on URI lookups~(e.g., \cite{%
DBLP:conf/semweb/HartigBF09%
,DBLP:conf/semweb/LadwigT10%
,DBLP:journals/www/UmbrichHKHP11%
})}%
. However, there exist other types of
	\removable{Web}
interfaces to access an RDF data source, including the Triple Pattern Fragment~(TPF) interface~\cite{verborgh2016triple}, the Bindings-Restricted TPF~(brTPF) interface~\cite{hartig2016brtpf}, the SaGe interface~\cite{DBLP:conf/www/MinierSM19}, and the smart-KG interface~\cite{DBLP:conf/www/AzzamFABP20}. Given the fact that
	each type of interface \removable{has particular properties and} makes
different trade-offs~\cite{%
DBLP:conf/www/AzzamFABP20%
,hartig2016brtpf%
,hartig2017formal%
,DBLP:conf/www/MinierSM19%
,verborgh2016triple%
}%
	\removable{~(and the same will hold for other types of such interfaces proposed in the future)}%
, any provider of an RDF data source may choose to offer a different type of
interface, which leads to
	federations that are heterogeneous in terms of these~interfaces. 

Such a heterogeneity poses extra challenges for query federation engines%
	---especially during query planning---%
because different interfaces may require~(or enable!) the engine to leverage specific physical operators, and not all forms of subqueries can be answered directly by every
interface. To the best of our knowledge, there does not exist any research on systematic approaches to tackle these challenges. This paper
	provides
a starting point for~such~research.

We argue that any principled approach to query such heterogeneous federations
	\removable{of RDF data sources}
has to be based on a solid formal foundation. This foundation should provide not only a formal data model that captures this notion of federations, including a corresponding query semantics, but also formal concepts to precisely define the artifacts produced\hidden{ and reasoned over} by the various steps of query planning. There are typically three main types of such artifacts in a query federation engine: the results of the query decomposition~\&~source selection step~\cite{DBLP:journals/tlsdkcs/VidalCAMP16}, logical query plans, and physical query plans.
%
	We observe that, so
far, such artifacts have been treated very informally in the literature on query engines for~(homogeneous) federations of RDF data. That is, the authors talk about query plans only in terms of \hidden{a few }examples, where these examples are typically informal illustrations that visually represent some form of a tree in which the leaf nodes are one or more triple patterns with various annotations and the internal nodes are operators from the SPARQL algebra~(sometimes in combination with some additional, in\-for\-mal\-ly-in\-tro\-duced operators~\cite{schwarte2011fedx}). The lack of approaches to describe query plans formally is the focus of our work in this paper.

\medskip\noindent
\textbf{Contributions%
	\removable{\ and organization of the paper}%
:}
After introducing a suitable formal data model%
	\ and a corresponding query semantics for using the query language SPARQL%
~(cf.\ Section~\ref{sec:DataModel}), we make our main conceptual contribution: We define a language, called \emph{FedQPL}, that can be used to describe logical query plans formally~(cf.\ Section~\ref{sec:Language}). This language can be applied both to define query planning and optimization approaches in a more precise manner and to actually represent the logical plans in a query engine%
.
	FedQPL features operators to explicitly capture the intention to execute a particular subquery at a specific federation member and to
distinguish whether such an access to a federation member is meant to be based solely on the given subquery or also on intermediate results obtained for other subqueries.
	We argue that such features
are paramount for any principled approach to query planning in heterogeneous federations where
	the characteristics and limitations of different
data access interfaces have to be taken into account%
	\removable{~(and, certainly, these features can also be leveraged when defining new approaches that focus on homogeneous federations)}%
.

	Given the full definition of FedQPL, we \removable{then} study a \removable{specific} fragment of this
language
	\removable{that can be used}
to describe which
	federation members
have to be contacted for which part of a given query%
%
	~(cf.\ Section~\ref{sec:SourceSelection})%
. Hence,
	this language fragment provides
a formal tool
	\removable{that is} needed
to develop well-de\-fined source selection approaches.
	Moreover, in practice, this language fragment can be used to represent the results of such approaches directly as \emph{initial logical query plans} that can then be optimized in subsequent query planning steps.
As technical contributions regarding this language fragment we show that the corresponding source selection problem is NP-hard and in $\Sigma_2^\mathrm{P}$%
	. Additionally,
we show
	that \removable{the language fragment cannot only capture any possible output of existing source selection approaches for homogeneous federations but,} even when used only for such federations, it can express solutions to the source selection problem that these approaches
are not able~to~produce.

Finally, as another technical contribution, we show a comprehensive set of
	equivalences for FedQPL expressions that can be used as query rewriting rules for query optimization~(cf.\ Section~\ref{sec:Equivalences}).

\medskip
\noindent
\textbf{Limitations:}
	It
is not the purpose of this paper to develop concrete approaches or techniques for source selection, query planning, or query optimization. Instead, we focus on providing a solid theoretical foundation for such work in the future. A specific limitation regarding our contributions
	in this paper
is that the
	given definition of FedQPL covers
only the join-union fragment of SPARQL. However,
	the formalism can easily be extended with additional operators%
%
	, which is part of our future work%
.

%% file: preliminaries.tex
\section{Preliminaries} \label{sec:Preliminaries}

This section
	\removable{provides a brief introduction of}
the concepts of RDF and SPARQL that are relevant for our work in this paper. Due to space limitations, we focus on
	\removable{notation and relevant symbols,}
and refer to the literature for
	the
detailed formal~definitions~\cite{Harris13:SPARQL1_1Language, DBLP:journals/tods/PerezAG09}. It is important to note that, in this paper, we focus on the set-based semantics of SPARQL as introduced by P{\'{e}}rez et al.~\cite{DBLP:journals/tods/PerezAG09}.

We assume
pairwise disjoint, countably infinite sets: $\symAllURIs$~(URIs), $\symAllBNodes$~(blank nodes), $\symAllLiterals$~(literals), and $\symAllVariables$~(variables%
).
	An \emph{RDF triple} is a tuple $(s,p,o) \in (\symAllURIs \cup \symAllBNodes) \times \symAllURIs \times (\symAllURIs \cup \symAllBNodes \cup \symAllLiterals)$.
%
	A set of such triples is called an \emph{RDF graph}.
%

The fundamental building block of SPARQL queries is the notion of a \emph{basic graph pattern}~(\emph{BGP})~\cite{Harris13:SPARQL1_1Language}; each such BGP is a
	nonempty
set of \removable{so-called} \emph{triple patterns} where every triple pattern is a tuple in $(\symAllVariables \cup \symAllURIs) \times (\symAllVariables \cup \symAllURIs) \times (\symAllVariables \cup \symAllURIs \cup \symAllLiterals)$.\footnote{For the sake of
	simplicity
we do not permit blank nodes in triple patterns. In practice, each blank node in a SPARQL query can be replaced by a new variable.}
%
	Other
types of graph patterns
	\removable{for SPARQL queries}
can be constructed by combining BGPs using various operators%
	\removable{~(e.g., UNION, FILTER\removable{, OPTIONAL})}%
~\cite{Harris13:SPARQL1_1Language}%
.

The result of evaluating any such graph pattern~$\symPattern$ over an RDF graph~$\symRDFgraph$ is a set---\removable{typically} denoted by~$\eval{\symPattern}{\symRDFgraph}$~\cite{DBLP:journals/tods/PerezAG09}---that consists of \removable{so-called}
	solution mappings, where \removable{a \emph{solution mapping}} is a partial function
%
	$\mu\!:\symAllVariables\rightarrow\symAllURIs \cup \symAllBNodes \cup \symAllLiterals$.
%
%

For the case that $\symPattern$ is a BGP~$\symBGP$, $\eval{\symPattern}{\symRDFgraph}$
	consists of every solution mapping~$\mu$ for which $\fctDom{\mu} = \fctVars{\symBGP}$ and $\mu[\symBGP] \subseteq \symRDFgraph$,
where $\fctVars{\symBGP}$ is the set of all variables in $\symBGP$ and $\mu[\symBGP]$ denotes the BGP that we obtain by replacing the variables in $\symBGP$ according to $\mu$; notice that $\mu[\symBGP]$ is a set of RDF triples if $\fctVars{\symBGP} \subseteq \fctDom{\mu}$.
	For any triple pattern~$\symTP$, we write $\eval{\symTP}{\symRDFgraph}$ as a shorthand notation for $\eval{\lbrace\symTP\rbrace}{\symRDFgraph}$.

For any more complex graph pattern~$\symPattern$, $\eval{\symPattern}{\symRDFgraph}$ is defined based on
	\removable{an algebra~\cite{DBLP:journals/tods/PerezAG09}. This algebra---which we call} the \emph{SPARQL algebra}\removable{---}%
consists of several
operators over sets of solution mappings, including $\Join$~(join) and $\cup$~(union)~\cite{DBLP:journals/tods/PerezAG09}. For instance, the join of
	two sets of solution mappings, $\Omega_1$~and~$\Omega_2$,
is defined as
	$\Omega_1 \Join \Omega_2 \definedAs \lbrace \mu_1 \cup \mu_2 \mid \mu_1 \in \Omega_1, \, \mu_2 \in \Omega_2, \text{ and } \mu_1 \sim \mu_2 \rbrace ,$
where we write $\mu_1 \sim \mu_2$
	if $\mu_1(?v) = \mu_2(?v)$ for every variable~$?v$ in $\fctDom{\mu_1} \cap \fctDom{\mu_2}$%
	\removable{; in this case, the combination of $\mu_1$ and $\mu_2$, denoted by $\mu_1 \cup \mu_2$, is also a solution mapping}%
.
It is not difficult to see that, for two BGPs $\symBGP_1$ and $\symBGP_2$, it holds that
	$\eval{\symBGP_1}{\symRDFgraph} \Join \eval{\symBGP_2}{\symRDFgraph} = \eval{\symBGP}{\symRDFgraph}$ where $\symBGP = \symBGP_1 \cup \symBGP_2$~\cite{DBLP:journals/tods/PerezAG09}.
Since the operators $\Join$ and $\cup$ are associative and commutative~\cite{DBLP:journals/tods/PerezAG09, DBLP:conf/icdt/Schmidt0L10}, we can avoid parenthesis when combining more than two sets of solution mappings using either of these operators; e.g., $(\Omega_1 \Join \Omega_2) \Join \Omega_3 = \Omega_1 \Join \Omega_2 \Join \Omega_3$.

%% file: datamodel.tex
\section{Data Model} \label{sec:DataModel}

This section introduces a
data model
	that captures the notion of
a federation
	of RDF data sources
that is heterogeneous in terms of the types of data access interfaces of the federation members.
	The
model includes a query semantics that defines the expected result of executing a SPARQL query over such~a~federation.

A key concept of the data model is that of an \emph{interface}, which we
	abstract
by
	two components:
a language for expressing data access requests and a function that, given an RDF graph, turns any expression in the language of the interface into a set of solution~mappings.

\begin{mydefinition} \label{def:Interface}
An
	\definedTerm{RDF data access interface}~(\definedTerm{interface})%
~$\iface$ is a tuple~$\ifaceTuple$
	where $\ifaceReqLang{}$ denotes a language and $\ifaceFct$ is a function
that maps every pair~$(\ifaceReqLangExp,\symRDFgraph)$, consisting of an expression~$\ifaceReqLangExp$ in $\ifaceReqLang{}$ and an RDF graph~$\symRDFgraph$, to a set of solution mappings.
\todo{TODO: add a paragraph about the responses (which are not necessarily sets of solution mappings in practice)}
\end{mydefinition}


The following
	\removable{three}
examples illustrate how
	any concrete interface can be defined \removable{in the context of our formalization}.

\begin{myexample} \label{ex:Interface:SPARQL}
The \emph{SPARQL endpoint interface}~\cite{feigenbaum2013sparql}
enables clients to request the
	result for
\emph{any} SPARQL query over the serv\-er-side dataset.
Hence, a server that provides this interface executes each requested SPARQL query over its dataset and, then, returns the result to the client. We may abstract this functionality by defining an interface $\iface_\mathsf{sparql} = (\ifaceReqLangSPARQL,\ifaceFct_\mathsf{sparql})$ where the expressions in $\ifaceReqLangSPARQL$ are all
	\removable{SPARQL}
graph patterns and $\ifaceFct_\mathsf{sparql}$ is defined
for every
graph pattern~$\symPattern$ and every RDF graph~$\symRDFgraph$%
	\ such
that $\ifaceFct_\mathsf{sparql}(\symPattern,\symRDFgraph) \definedAs \eval{\symPattern}{\symRDFgraph}$.
\end{myexample}

\begin{myexample} \label{ex:Interface:TPF}
While a SPARQL endpoint server enables clients to query its dataset by using the full expressive power of SPARQL, providing such a comparably complex functionality may easily overload such servers~\cite{DBLP:conf/semweb/ArandaHUV13}. To address this issue, several more restricted types of interfaces have been proposed with the goal to
	shift some of the query execution effort from the server to the clients.
The first of these alternatives has been the \emph{Triple Pattern Fragment (TPF) interface}~\cite{verborgh2016triple} that allows clients
	only to send triple pattern queries to the server. More specifically, via this interface, clients can
request the triples from the serv\-er-side dataset that match a given triple pattern.
In terms of our model, we abstract this functionality by an interface~$\iface_\mathsf{TPF} = (\ifaceReqLangTPF,\ifaceFct_\mathsf{TPF})$ where the expressions in $\ifaceReqLangTPF$ are all triple patterns and $\ifaceFct_\mathsf{TPF}$ is defined
for every triple pattern~$\symTP$ and every RDF graph~$\symRDFgraph$%
	\ such
that $\ifaceFct_\mathsf{TPF}(\symTP,\symRDFgraph) \definedAs \eval{\symTP}{\symRDFgraph}$.
\end{myexample}

\begin{myexample} \label{ex:Interface:brTPF}
The \emph{Bindings-Restricted
	TPF
(brTPF) interface}~\cite{hartig2016brtpf} extends the TPF interface by allowing clients to optionally attach intermediate results to triple pattern requests. The response to such a request is expected to contain
	only those matching triples
that are guaranteed to contribute in a join with the given intermediate result.
	We define a corresponding
interface $\iface_\mathsf{brTPF} = (\ifaceReqLangBrTPF,\ifaceFct_\mathsf{brTPF})$ where the expressions in $\ifaceReqLangBrTPF$ are
	i)~all triple patterns and ii)~%
all pairs consisting of a triple pattern and a set of solution mappings, and $\ifaceFct_\mathsf{brTPF}$ is defined
for every triple pattern~$\symTP$, every set~$\Omega$ of solution mappings, and every RDF graph~$\symRDFgraph$%
	\ such
that%
\begin{align*}
	\ifaceFct_\mathsf{brTPF}\bigl( \symTP,\symRDFgraph \,\bigr) &\definedAs \eval{\symTP}{\symRDFgraph}, \text{ and}
	\\
	\ifaceFct_\mathsf{brTPF}\bigl( (\symTP,\Omega),\symRDFgraph \,\bigr) &\definedAs
	\begin{cases}
		\eval{\symTP}{\symRDFgraph} & \text{if
					$\Omega = \emptyset$,} \\
		\{ \mu \in \eval{\symTP}{\symRDFgraph} \mid
			\exists\mu'\!\in\Omega: \mu \sim \mu'  \} & \text{else} .
	\end{cases}
\end{align*}
\end{myexample}

Given the notion of an interface, we can now define our notion of a federation.

\begin{mydefinition} \label{def:Federation}
A \definedTerm{federation member}~$\member$ is a pair~$(\symRDFgraph,\iface)$
	that consists of an RDF graph~$\symRDFgraph$ and an interface~$\iface$.
A \definedTerm{federation}~$\Fed$ is a
	finite and nonempty
set of federation members such that every member~$\member = (\symRDFgraph,\iface)$ in $\Fed$ uses a disjoint set of blank nodes; i.e., $\fctBNodes{\symRDFgraph} \cap \fctBNodes{\symRDFgraph'} = \emptyset$ for every other
$\member'\! = (\symRDFgraph'\!,\iface')$ in~$\Fed$%
.
\end{mydefinition}

\begin{myexample} \label{ex:Federation}
As a running example
	\removable{for this paper},
we consider a~%
	\removable{simple}
federation $\Fed_\mathsf{ex} = \lbrace \member_1, \member_2, \member_3 \rbrace$
	with
three members: $\member_1 = (\symRDFgraph_1,\iface_\mathsf{brTPF})$, $\member_2 = (\symRDFgraph_2,\iface_\mathsf{TPF})$, and $\member_3 = (\symRDFgraph_3,\iface_\mathsf{sparql})$.
The data in
	these members
are the following RDF graphs%
%
.%
\begin{align*}
G_1 = \lbrace &
	(\rdfterm{a}, \rdfterm{foaf\!:\!knows}, \rdfterm{c})
\rbrace
\quad \quad
G_2 = \lbrace
	(\rdfterm{c}, \rdfterm{foaf\!:\!name}, \rdfterm{"Lee"}),
\\[-1mm]
& \hspace{38.7mm}
	(\rdfterm{d}, \rdfterm{foaf\!:\!name}, \rdfterm{"Alice"})  
\rbrace
\\
G_3 = \lbrace &
	(\rdfterm{a}, \rdfterm{foaf\!:\!knows}, \rdfterm{b}),
	(\rdfterm{b}, \rdfterm{foaf\!:\!name}, \rdfterm{"Peter"})
\rbrace
\end{align*}
\end{myexample}

Informally, the result of
	a
SPARQL query over
	\removable{such}
a federation should be the same as if the query was executed over the union of all the RDF data available in all the federation members. Formally,
	this
query semantics is defined as follows.

\begin{mydefinition} \label{def:QuerySemantics}
The \definedTerm{evaluation} of a SPARQL graph pattern~$\symPattern$ over a federation~$\Fed$, denoted by $\eval{\symPattern}{\Fed}$, is a set of solution mappings that is defined
	as:
$\eval{\symPattern}{\Fed} \definedAs \eval{\symPattern}{\symRDFgraph_\mathrm{union}}$
where $\symRDFgraph_\mathrm{union} = \bigcup_{(\symRDFgraph,\iface) \in \Fed} \symRDFgraph$%
	\removable{~~~(and $\eval{\symPattern}{\symRDFgraph_\mathrm{union}}$ is as in Section~\ref{sec:Preliminaries})}%
.
\end{mydefinition}

\begin{myexample} \label{ex:QuerySemantics}
Consider a BGP~$\symBGP_\mathsf{ex} = \lbrace \symTP_1, \symTP_2 \rbrace$
	with
the two triple patterns
$\symTP_1 = ( ?x, \rdfterm{foaf\!:\!knows}, ?y )$ and
$\symTP_2 = ( ?y, \rdfterm{foaf\!:\!name}, ?z )$.
When evaluating
	$\symBGP_\mathsf{ex}$
over our example federation~$\Fed_\mathsf{ex}$ in Example~\ref{ex:Federation},
	we obtain
$\eval{\symBGP_\mathsf{ex}}{\Fed_\mathsf{ex}} = \lbrace \mu_1,\mu_2 \rbrace$ with
	\begin{align*} \mu_1 & = \lbrace ?x \!\rightarrow\! \rdfterm{a}, ?y \!\rightarrow\! \rdfterm{c}, ?z \!\rightarrow\! \rdfterm{"Lee"} \rbrace , \text{ and} \\ \mu_2 & = \lbrace ?x \!\rightarrow\! \rdfterm{a}, ?y \!\rightarrow\! \rdfterm{b}, ?z \!\rightarrow\! \rdfterm{"Peter"} \rbrace . \end{align*}
\end{myexample}

	For any
interface~$\iface=\ifaceTuple$, we say that $\iface$ \emph{supports triple pattern requests} if we can write every triple pattern~$\symTP$ as a request~$\ifaceReqLangExp$ in $\ifaceReqLang{}$ such that for every RDF graph~$\symRDFgraph$ we have that $\ifaceFct(\ifaceReqLangExp,\symRDFgraph)=\eval{\symTP}{\symRDFgraph}$.
Similarly, $\iface$ \emph{supports BGP requests} if every BGP~$\symBGP$ can be written as a request~$\ifaceReqLangExp$ in $\ifaceReqLang{}$ such that for every RDF graph~$\symRDFgraph$ we have that $\ifaceFct(\ifaceReqLangExp,\symRDFgraph)=\eval{\symBGP}{\symRDFgraph}$.
	\removable{Then,} out
of the three aforementioned~interfaces~(cf.\ Examples~\ref{ex:Interface:SPARQL}--\ref{ex:Interface:brTPF}), only the SPARQL endpoint interface supports BGP requests, but all three support triple pattern~requests.

The notion of support for triple pattern requests~(resp.\ BGP requests)
	can be carried
over to federation members%
	; e.g., if
the interface~$\iface$ of a federation member $\member = (\symRDFgraph,\iface)$ supports triple pattern requests%
, we also say that $\member$ \emph{supports triple pattern~requests}%
.

	Finally, we say that a
federation%
\ is \emph{triple pattern accessible} if
	all of its members support
triple pattern~requests%
. For the complexity results in this paper~(cf.\ Section~\ref{sec:SourceSelection:Complexity}) we assume that such federations can be encoded on the tape of a Turing Machine such that all triples that match a given triple pattern can be found in polynomial time.

%% file: language.tex
\section{Query Plan Language} \label{sec:Language}

Now we
introduce our language, FedQPL, to describe logical plans for executing queries over heterogeneous federations of RDF data.

	A
logical plan is a tree of algebraic operators that capture a declarative notion of how their output is related to their input~(rather than a concrete algorithm of how the output will be produced, which is the focus of physical plans).
	Additionally, a logical plan typically also indicates an order in which the operators will be evaluated.
%
	\par As mentioned in the introduction,
		such plans are presented only informally in existing work on query engines for\hidden{~(homogeneous)} federations of RDF data.
%
In contrast to these informal representations, we define
	\removable{both the syntax and the semantics of}
FedQPL formally.
	In comparison to the standard SPARQL algebra, the main innovations of FedQPL
are that it contains operators to make explicit which federation member is accessed in each part of a query plan and to distinguish different ways of accessing a federation~member.

We begin by defining the syntax of FedQPL.

\begin{mydefinition} \label{def:FedQPL:Syntax}
A \definedTerm{FedQPL expression} is an expression~$\varphi$ that can be constructed from the following grammar,
	in which
$\xxxAlgebraOperatorName{req}$,
$\xxxAlgebraOperatorName{tpAdd}$,
$\xxxAlgebraOperatorName{bgpAdd}$,
$\xxxAlgebraOperatorName{join}$,
$\xxxAlgebraOperatorName{union}$,
$\xxxAlgebraOperatorName{mj}$,
$\xxxAlgebraOperatorName{mu}$,
$($, and $)$ are terminal symbols, 
$\ifaceReqLangExp$ is an expression in the request language~$\ifaceReqLang{}$ of some interface,
$\member$ is a federation member,
$\symTP$ is a triple pattern,
$\symBGP$ is a BGP,
and $\mathit{\Phi}$ is a nonempty set of FedQPL expressions.
\begin{align*}
\varphi ::= \quad &
      \request{\ifaceReqLangExp}{\member}
\quad | \quad
      \tpAdd{\varphi}{\symTP}{\member}
\quad | \quad
      \bgpAdd{\varphi}{\symBGP}{\member}
\quad | \quad
\\ &
      \join{\varphi}{\varphi}
\quad | \quad
      \union{\varphi}{\varphi}
\quad | \quad
       \xxxAlgebraOperatorName{mj} \, \mathit{\Phi}
\quad | \quad
       \xxxAlgebraOperatorName{mu} \, \mathit{\Phi}
\end{align*}
\end{mydefinition}

Before we present the formal semantics of FedQPL expressions, we provide an intuition of the different operators
	of the language.

\subsection{Informal Overview and Intended Use}

The
	first operator, $\xxxAlgebraOperatorName{req}$,
	captures the \removable{intention} to retrieve
the result of a given (sub)query from a given federation member, where the (sub)query is expressed in the request language of the interface provided by the federation member.

\begin{myexample} \label{ex:request}
For BGP~%
	$\symBGP_\mathsf{ex} = \lbrace \symTP_1, \symTP_2 \rbrace$
in Example~\ref{ex:QuerySemantics} we observe that member $\member_1$ of our example federation~%
	$\Fed_\mathsf{ex}$%
~(cf.\ Example~\ref{ex:Federation}) can contribute a solution mapping for~$\symTP_1$, whereas $\member_2$ can contribute two solution mappings for $\symTP_2$. The intention to retrieve these solution mappings from these federation members can be represented \removable{in a logical plan} by the operators $\request{\symTP_1}{\member_1\!}$ and $\request{\symTP_2}{\member_2\!}$, respectively.
%
	\end{myexample}\begin{myexample} \label{ex:request2} We also observe that, by accessing
	\removable{federation member~}%
$\member_3 \in \Fed_\mathsf{ex}$, we may retrieve a nonempty result for the
	example BGP~$\symBGP_\mathsf{ex}$ as a whole.
The intention to do so can be represented by the operator $\request{\lbrace\symTP_1,\symTP_2\rbrace\!}{\member_3}$. Notice that such a request
	with a BGP~(rather than with a single triple pattern)
is possible only because the interface of $\member_3$ is the SPARQL endpoint interface~$\iface_\mathsf{sparql}$~(cf.\ Example~\ref{ex:Federation}) which~supports~BGP~requests.
\end{myexample}

The second operator, $\xxxAlgebraOperatorName{tpAdd}$, captures the intention to access a\hidden{~(triple pattern accessible)} federation member to obtain solution mappings for a single triple pattern that need to be compatible with~(and are to be joined with) solution mappings in a given intermediate query result.

\begin{myexample} \label{ex:tpAdd}
Continuing with our example BGP~%
(cf.\ Example~\ref{ex:QuerySemantics}) over our example federation~%
	$\Fed_\mathsf{ex}$%
~(cf.\ Example~\ref{ex:Federation}), we observe that
	the solution mapping for $\symTP_1$ from $\member_1$ can be joined with only one of the solution mappings for $\symTP_2$ from $\member_2$.
To produce the join between the two sets of solution mappings~(i.e., between $\eval{\symTP_1}{G_1}$ and $\eval{\symTP_2}{G_2}$) we may use
	the set $\eval{\symTP_1}{G_1}$ 
as input to retrieve only
	those solution mappings for $\symTP_2$ from $\member_2$ that can actually be joined with the solution mapping in $\eval{\symTP_1}{G_1}$.
The plan to do so can be represented by
	combining a $\xxxAlgebraOperatorName{tpAdd}$ operator with the $\xxxAlgebraOperatorName{req}$ operator that retrieves $\eval{\symTP_1}{G_1}$,
which gives us the following FedQPL expression.
\begin{equation*}
\tpAddB{ \request{\symTP_1}{\member_1\!} }{\symTP_2}{\member_2\!}
\end{equation*}
\end{myexample}

While the third operator, \xxxAlgebraOperatorName{bgpAdd}, represents a BGP-based variation of $\xxxAlgebraOperatorName{tpAdd}$, the remaining operators
	($\xxxAlgebraOperatorName{join}$, $\xxxAlgebraOperatorName{union}$, $\xxxAlgebraOperatorName{mj}$, and $\xxxAlgebraOperatorName{mu}$)
lift the standard SPARQL algebra operators join and union into the FedQPL language. In particular, $\xxxAlgebraOperatorName{join}$ is a binary operator that joins two inputs whereas $\xxxAlgebraOperatorName{mj}$ represents a multiway variation of a join that can combine an arbitrary number of inputs. In contrast to $\xxxAlgebraOperatorName{tpAdd}$ and $\xxxAlgebraOperatorName{bgpAdd}$, the operators $\xxxAlgebraOperatorName{join}$ and $\xxxAlgebraOperatorName{mj}$ capture the intention to obtain the input sets of solution mappings independently and, then, join them only in the query federation engine alone. The operators $\xxxAlgebraOperatorName{union}$ and $\xxxAlgebraOperatorName{mu}$ are the union-based counterparts of $\xxxAlgebraOperatorName{join}$ and $\xxxAlgebraOperatorName{mj}$. Our language contains both the binary and the multiway variations of these operators to allow for query plans in which the intention to apply a multiway algorithm can be
	distinguished explicitly from the intention to use some algorithm designed for the binary case;
additionally, the operators $\xxxAlgebraOperatorName{mj}$ and $\xxxAlgebraOperatorName{mu}$ can be used during early stages of query planning when the order in which multiple intermediate results will be combined is not decided yet.

\begin{myexample} \label{ex:union}
By executing the operator $\request{\lbrace\symTP_1,\symTP_2\rbrace\!}{\member_3}$ as discussed in Example~\ref{ex:request2}, we can obtain a
	solution mapping that is part of the query result for our
example BGP~$\symBGP_\mathsf{ex}$~(cf.\ Example~\ref{ex:QuerySemantics}).
	Another such part of this result may be obtained based on the FedQPL expression in Example~\ref{ex:tpAdd}.
	These partial results can then be combined by using the $\xxxAlgebraOperatorName{union}$ operator as follows.
\begin{equation*}
\unionB{\; \request{\lbrace\symTP_1,\symTP_2\rbrace\!}{\member_3} \; }{ \; \tpAddB{ \request{\symTP_1}{\member_1\!} }{\symTP_2}{\member_2\!} \;}
\end{equation*} 
\end{myexample}


To further elaborate on the distinction between the $\xxxAlgebraOperatorName{join}$ operator~(as well as its multiway counterpart $\xxxAlgebraOperatorName{mj}$) and the operators $\xxxAlgebraOperatorName{tpAdd}$ and $\xxxAlgebraOperatorName{bgpAdd}$ we emphasize that the latter can be used in cases in which the processing power of a federation member can be exploited to
	join an input set of solution mappings with
the result of evaluating a triple pattern~(or a BGP) over the data of that federation member. Specific algorithms
that can be used as implementations of $\xxxAlgebraOperatorName{tpAdd}$ (or $\xxxAlgebraOperatorName{bgpAdd}$) in such cases are RDF-specific variations of the semijoin~\cite{DBLP:journals/tods/BernsteinGWRR81} and the bind join~\cite{DBLP:conf/vldb/HaasKWY97}. Concrete examples of such algorithms
	can be found in the SPARQL endpoint federation engines FedX~\cite{schwarte2011fedx}, \hidden{SPLENDID~\cite{gorlitz2011splendid}, }SemaGrow~\cite{charalambidis2015semagrow}, and CostFed~\cite{saleem2018costfed}, as well as in the brTPF client~\cite{hartig2016brtpf}.
Such algorithms rely on a data access interface in which the given input solution mappings can be captured as part of the requests.
However,
	for less expressive interfaces~(such as the TPF interface),
the $\xxxAlgebraOperatorName{tpAdd}$ operator can also be implemented using
	a variation of an index nested-loops join
in which
	a separate request is created for each input solution mapping%
~%
	\cite{quilitz2008querying,vidal2010efficiently,verborgh2016triple}.
In contrast, a standard~(local) nested-loops join---which has also been proposed in the literature on
	SPARQL federation engines%
~\cite{quilitz2008querying}---would be an implementation of the $\xxxAlgebraOperatorName{join}$ operator.
	Further examples of join algorithms that have been proposed for such engines and that would be implementations of
the $\xxxAlgebraOperatorName{join}$ operator are a group join~\cite{vidal2010efficiently}, a simple hash join~\cite{acosta2011anapsid,gorlitz2011splendid}, a symmetric hash join~\cite{acosta2011anapsid,saleem2018costfed}, and a~merge~join~\cite{charalambidis2015semagrow}.

\subsection{Validity}
While the various FedQPL operators can be combined arbitrarily as per the grammar in Definition~\ref{def:FedQPL:Syntax}, not every operator can be used arbitrarily for every federation member. In contrast,
	as already indicated in Example~\ref{ex:request2},
depending on their interface, federation members may not be capable to be accessed in the way as required by a particular operator. This observation leads to
	a
notion of validity of FedQPL
	expressions \removable{that we define recursively as follows}.

\begin{mydefinition} \label{def:FedQPL:Validity}
Let $\Fed$ be a federation.
A FedQPL expression~\definedTerm{$\varphi$ is valid for $\Fed$} if it has the following properties:
\begin{enumerate}
	\itemsep2mm 

	\item if $\varphi$ is of the form $\request{\ifaceReqLangExp}{\member}$ where $\member = (\symRDFgraph,\iface)$ with $\iface=\ifaceTuple$, then $\member \in \Fed$ and $\ifaceReqLangExp$ is an expression in $\ifaceReqLang{}$;

	\item if $\varphi$ is of the form $\tpAdd{\varphi'}{\symTP}{\member}$, then $\member \in \Fed$, $\member$ supports triple pattern requests, and $\varphi'$ is valid;

	\item if $\varphi$ is of the form $\bgpAdd{\varphi'}{\symBGP}{\member}$, then $\member \in \Fed$, $\member$ supports BGP requests, and $\varphi'$ is valid;

	\item if $\varphi$ is of the form $\join{\varphi_1}{\varphi_2}$, then $\varphi_1$ and $\varphi_2$ are valid;

	\item if $\varphi$ is of the form $\union{\varphi_1}{\varphi_2}$, then $\varphi_1$ and $\varphi_2$ are valid;

	\item if $\varphi$ is of the form $\xxxAlgebraOperatorName{mj} \mathit{\Phi}$, then every $\varphi'\! \in \mathit{\Phi}$ is valid;

	\item if $\varphi$ is of the form $\xxxAlgebraOperatorName{mu} \mathit{\Phi}$, then every $\varphi'\! \in \mathit{\Phi}$ is valid.
\end{enumerate}
\end{mydefinition}

\begin{myexample} \label{ex:Validity}
All FedQPL expressions presented in the previous examples are valid for our example federation~$\Fed_\mathsf{ex}$ in Example~\ref{ex:Federation}.
\end{myexample}

\begin{myexample} \label{ex:Validity2}
An expression that is not valid for $\Fed_\mathsf{ex}$ is $\request{\lbrace\symTP_1,\symTP_2\rbrace\!}{\member_1}$. The issue with this
	expression---and with every other expression that contains it as a subexpression---is
that it assumes that federation member~$\member_1$ supports BGP requests, which is not the case because the interface of $\member_1$ is $\iface_\mathsf{brTPF}$.
For the same reason, $\member_1$ cannot be used in the $\xxxAlgebraOperatorName{bgpAdd}$ operator, which makes expressions such as
	$\bgpAdd{ \varphi }{\lbrace\symTP_1,\symTP_2\rbrace\!}{\member_2}$ to be invalid as well (for any subexpression~$\varphi$).
\end{myexample}

In the remainder of this paper we assume FedQPL expressions that are valid for the federation for which they have been created. 

\subsection{Semantics}
Now we are ready to define a formal semantics of FedQPL. To this end, we introduce a function that defines for each~(valid) FedQPL expression, the set of solution mappings that is expected as the result of evaluating the expression. 

\begin{mydefinition} \label{def:FedQPL:Semantics}
Let $\varphi$ be a FedQPL expression that is valid for a federation~$\Fed$. The \definedTerm{solution mappings obtained with $\varphi$}, denoted by $\sols{\varphi}$, is a set of solution mappings that is defined
	recursively:
\begin{enumerate}
	\itemsep2mm 
	\item if $\varphi$ is of the form $\request{\ifaceReqLangExp}{\member}$ where $\member = (\symRDFgraph,\iface)$ with $\iface=\ifaceTuple$, then $\sols{\varphi} \definedAs \ifaceFct(\ifaceReqLangExp,\symRDFgraph)$;

	\item if $\varphi$ is
	$\tpAdd{\varphi'}{\symTP}{\member}$ where $\member = (\symRDFgraph,\iface)$ with $\iface=\ifaceTuple$, then $\sols{\varphi} \definedAs \sols{\varphi'} \Join \ifaceFct(\symTP,\symRDFgraph)$;

	\item if $\varphi$ is
	$\bgpAdd{\varphi'}{\symBGP}{\member}$ where $\member = (\symRDFgraph,\iface)$ with $\iface=\ifaceTuple$, then $\sols{\varphi} \definedAs \sols{\varphi'} \Join \ifaceFct(\symBGP,\symRDFgraph)$;

	\item if $\varphi$ is
	$\join{\varphi_1}{\varphi_2}$, then $\sols{\varphi} \definedAs \sols{\varphi_1} \Join \sols{\varphi_2}$;

	\item if $\varphi$ is
	$\union{\varphi_1}{\varphi_2}$, then $\sols{\varphi} \definedAs \sols{\varphi_1} \cup \sols{\varphi_2}$;

	\item if $\varphi$ is
		of the form
	$\xxxAlgebraOperatorName{mj} \mathit{\Phi}$ where $\mathit{\Phi} = \lbrace \varphi_1, \varphi_2, ...\,, \varphi_n \rbrace$, then $\sols{\varphi} \definedAs \sols{\varphi_1} \Join \sols{\varphi_2} \Join ... \Join \sols{\varphi_n}$;

	\item if $\varphi$ is
		of the form
	$\xxxAlgebraOperatorName{mu} \mathit{\Phi}$ where $\mathit{\Phi} = \lbrace \varphi_1, \varphi_2, ...\,, \varphi_n \rbrace$, then $\sols{\varphi} \definedAs \sols{\varphi_1} \cup \sols{\varphi_2} \cup ... \cup \sols{\varphi_n}$.
\end{enumerate}
\end{mydefinition}

\begin{myexample} \label{ex:FedQPL:Semantics}
Given our example federation~%
	$\Fed_\mathsf{ex}$%
~(cf.\ Example~\ref{ex:Federation}), for the FedQPL expressions in Examples~\ref{ex:request2} and~\ref{ex:tpAdd} we have that
$$\sols{ \request{\lbrace\symTP_1,\symTP_2\rbrace\!}{\member_3}\, } = \lbrace \mu_2 \rbrace 
\quad\text{ and }\quad
\sols{ \tpAddB{ \request{\symTP_1}{\member_1\!} }{\symTP_2}{\member_2\!} } = \lbrace \mu_1 \rbrace ,$$
where the solution mappings $\mu_1$ and $\mu_2$ are as given in Example~\ref{ex:QuerySemantics}.
Consequently, for the expression in Example~\ref{ex:union} we thus have that
$$\sols{\; \unionB{ \request{\lbrace\symTP_1,\symTP_2\rbrace\!}{\member_3} \, }{ \tpAddB{ \request{\symTP_1}{\member_1\!} }{\symTP_2}{\member_2\!} } \;} = \lbrace \mu_1, \mu_2 \rbrace .$$
\end{myexample}

\subsection{Correctness}
Given that FedQPL expressions are meant to represent (logical) query execution plans to produce the result of a given BGP over a given federation, we also introduce a correctness property to indicate whether a FedQPL expression correctly captures a given BGP for a given federation. Informally, a FedQPL expression has this correctness property if the
	set of solution mappings
obtained with the expression is the result expected for the BGP over
	the~federation.

\begin{mydefinition} \label{def:FedQPL:Correctness}
Let $\symBGP$ be a BGP and $\Fed$ be a federation.
A FedQPL expression~$\varphi$ \definedTerm{is correct for $\symBGP$ over $\Fed$} if $\varphi$ is valid for $\Fed$ and it holds that $\sols{\varphi} = \eval{\symBGP}{\Fed}$.
\end{mydefinition}

\begin{myexample} \label{ex:FedQPL:Correctness}
Based on Examples~\ref{ex:QuerySemantics} and~\ref{ex:FedQPL:Semantics}, we can see that the FedQPL expression in Example~\ref{ex:union} is correct for our example BGP~$\symBGP_\mathsf{ex}$ over our example federation~$\Fed_\mathsf{ex}$.
\end{myexample}

This completes the definition of FedQPL. In the remainder of the paper we first focus on a fragment of the language that can be used to
	capture
the output of the query decomposition~\&~source selection step. Thereafter, we show equivalences for FedQPL expressions, which provide a formal foundation for logical query optimization.

%% file: srcselection.tex
\section{Source Selection and Initial Plans} \label{sec:SourceSelection}

An important aspect
	and one of the first steps
of planning
	the
execution of queries over a federation is to identify which federation members have to be contacted for which part of a given query. The key tasks of this step are referred to as
\emph{query decomposition} and \emph{source selection}~%
	\removable{\cite{DBLP:series/sci/GorlitzS11, DBLP:journals/ker/OguzEYDH15, DBLP:journals/tlsdkcs/VidalCAMP16, DBLP:reference/bdt/AcostaHS19}}.
In this section, we identify a fragment of FedQPL that can be used to capture the output of this step formally. We call the expressions in this fragment \emph{source assignments}. After defining them, we study their expressive power as well as the complexity of finding minimal source assignments.

We emphasize that
	this source assignments fragment of FedQPL provides
a
	foundation
to
	define query decomposition and source selection approaches formally, and to compare such approaches systematically.
Moreover, and perhaps more importantly from a practical perspective, if the output of such an approach is described in the form of a source assignment, it can readily be used as an \emph{initial logical plan} that can then be rewritten and refined during
	the
subsequent query optimization~steps.

\subsection{Source Assignments} \label{ssec:SourceAssignments}


The goal of query decomposition is to split a given BGP into smaller components, called \emph{subqueries}, which may be subsets of the BGP or even individual triple patterns. The goal of source selection is to assign to each such subquery the federation members from which we may retrieve a nonempty result for the subquery. We may capture such an assignment of a federation member to a subquery by the FedQPL operator $\xxxAlgebraOperatorName{req}$ for which we only have to consider requests in the form of a BGP or a triple pattern. Additionally, it needs to be specified how these individual assignments belong together such that the intermediate results that may be obtained from them can be combined correctly into the complete result of the given BGP. To this end, we may use the operators $\xxxAlgebraOperatorName{mu}$~(for intermediate results that cover the same subqueries) and $\xxxAlgebraOperatorName{mj}$~(for intermediate results of different subqueries). We emphasize that we select the multiway versions of union and join on
	purpose; they do not prescribe any order over their input operators, which captures more accurately the output of the query decomposition and source selection step (deciding on such an order is not part of this step, but of the subsequent query optimization steps).
With this, we have all
	parts
of
	the language fragment needed for source assignments.

\begin{mydefinition} \label{def:SourceAssignment}
A \definedTerm{source assignment} is a FedQPL expression that
	uses
only the operators 
$\xxxAlgebraOperatorName{req}$,
$\xxxAlgebraOperatorName{mj}$ and
$\xxxAlgebraOperatorName{mu}$,
and for each subexpression of the form $\request{\ifaceReqLangExp}{\member}$ it holds that $\ifaceReqLangExp$ is a triple pattern or a BGP.
\end{mydefinition}

\begin{myexample} \label{ex:SourceAssignment}
From the running example we recall that members $\member_1$ and $\member_3$ of our example federation~%
	$\Fed_\mathsf{ex}$%
	\ 
can contribute matching triples for~$\symTP_1$ of the example BGP~%
	$\symBGP_\mathsf{ex}$%
~(cf.\ Example~\ref{ex:QuerySemantics}), whereas $\symTP_2$ can be matched only in the data of $\member_2$ and $\member_3$. Hence, we may use the following source assignment~$a_\mathsf{ex}$ for $\symBGP_\mathsf{ex}$ over $\Fed_\mathsf{ex}$.
\begin{equation*}
\mjoinB{
\;
	\munion{
		\request{\symTP_1}{\member_1},
		\request{\symTP_1}{\member_3}
	}
\;,
\;
	\munion{
		\request{\symTP_2}{\member_2},
		\request{\symTP_2}{\member_3}
	}
\;
}
\end{equation*}
However,
the matching triple for $\symTP_1$ in $\member_1$ can be combined only with the triple for $\symTP_2$ in $\member_2$ and, similarly, the triple for $\symTP_1$ in $\member_3$ can be combined only with the triple for $\symTP_2$ in $\member_3$.
	Hence,
a more sophisticated
	query decomposition \&
source selection approach may prune one access to $\member_3$ by combining $\symTP_1$ and $\symTP_2$ for a single access to $\member_3$, which gives us the following source~assignment~$a'_\mathsf{ex}$%
.%
\begin{equation*}
\munionB{
\;
	\mjoin{
		\request{\symTP_1}{\member_1},
		\request{\symTP_2}{\member_2}
	}
\;,
\;
	\request{\lbrace\symTP_1,\symTP_2\rbrace}{\member_3}
\;
}
\end{equation*}
It can be easily verified that both of these source assignments, $a_\mathsf{ex}$ and $a'_\mathsf{ex}$, are correct for $\symBGP_\mathsf{ex}$ over $\Fed_\mathsf{ex}$.
\end{myexample}

\subsection{Exhaustive Source Assignments}

In the previous example~(Example~\ref{ex:SourceAssignment}) we assume to have detailed knowledge of the data available at all the federation members. In practice, however, this knowledge may be much more limited and is captured in some form of pre-populated data catalog with metadata about the federation members~\cite{DBLP:series/sci/GorlitzS11, DBLP:journals/ker/OguzEYDH15, DBLP:reference/bdt/AcostaHS19}. The particular types of metadata vary for each source selection approach that relies on such a catalog. On the other hand, there are also source selection approaches that do not use a pre-populated data catalog at all but, instead, aim to obtain some information about the data of the federation members during the source selection process~itself~\cite{schwarte2011fedx, DBLP:journals/tlsdkcs/VidalCAMP16, abdelaziz2017lusail}.

A straightforward approach
	that does not require any such information is to create a source assignment that requests
every triple pattern of the BGP separately at each federation member. Then, the results of these requests can be unioned per triple pattern and, finally, all the triple pattern specific unions can be
	joined.
We call a source assignment that captures this approach \emph{exhaustive}.

\begin{mydefinition} \label{def:ExhaustiveSourceAssignment}
Let $\Fed = \lbrace \member_1, \member_2, ...\, , \member_m \rbrace$ be a federation that is triple pattern accessible, and let $\symBGP = \lbrace \symTP_1, \symTP_2, ...\, , \symTP_n \rbrace$ be a BGP.
The \definedTerm{exhaustive source assignment} for $\symBGP$ over~$\Fed$ is the following source assignment.
$$
\mjoinB{
	\;
	\munion{
		\request{\symTP_1}{\member_1},...\,,\request{\symTP_1}{\member_m}
	}
	\;
	, ...\,,
	\;
	\munion{
		\request{\symTP_n}{\member_1},...\,,\request{\symTP_n}{\member_m}
	}
	\;
}
$$
\end{mydefinition}

The following result follows
	readily
from Definitions~\ref{def:QuerySemantics}, \ref{def:FedQPL:Semantics}, \ref{def:FedQPL:Correctness},
\ref{def:ExhaustiveSourceAssignment}.

\begin{myproposition} \label{prop:CorrectnessOfExhaustiveSourceAssignments}
Let $\symBGP$ be a BGP and $\Fed$ be a federation that is triple pattern accessible.
The exhaustive source assignment for $\symBGP$ over~$\Fed$ is correct for $\symBGP$ over~$\Fed$.
\end{myproposition}

\begin{mycorollary}
By Proposition~\ref{prop:CorrectnessOfExhaustiveSourceAssignments}, it follows
	trivially
that for every BGP~$\symBGP$ and every triple pattern accessible federation~$\Fed$, there exists a source assignment that is correct for $\symBGP$ over~$\Fed$.
\end{mycorollary}

 \begin{table*}[t!]
 \centering
 \begin{tabular}{llc}
 \hline
 \multicolumn{1}{c}{federation engine} & \multicolumn{1}{c}{source assignment} & sa-cost \\ \hline \hline
  FedX~\cite{schwarte2011fedx}
  &
$\mjoinB{
	\;
	\request{\lbrace\symTP_1,\symTP_2\rbrace}{\member_\mathsf{drb}}
	\, , \;
	\request{\lbrace\symTP_3,\symTP_4\rbrace}{\member_\mathsf{kegg}}
	\, , \;
	\munion{
		\request{\symTP_5}{\member_\mathsf{drb}}, \,
		\request{\symTP_5}{\member_\mathsf{kegg}}, \,
		\request{\symTP_5}{\member_\mathsf{dbp}}
	}
	\;
}$
 & 5 \\ \hline
  SemaGrow~\cite{charalambidis2015semagrow}
 &
$\mjoinB{
	\;
	\request{\lbrace\symTP_1,\symTP_2\rbrace}{\member_\mathsf{drb}}
	\, , \;
	\request{\symTP_3}{\member_\mathsf{kegg}}
	\, , \;
	\munion{
		\request{\symTP_4}{\member_\mathsf{kegg}}, \,
		\request{\symTP_4}{\member_\mathsf{chebi}}
	}
	\, , \;
	\munion{
		\request{\symTP_5}{\member_\mathsf{kegg}}, \,
		\request{\symTP_5}{\member_\mathsf{chebi}}
	}
	\;
}$
 & 6 \\ \hline
 CostFed~\cite{saleem2018costfed}
  &
$\mjoinB{
	\;
	\request{\lbrace\symTP_1,\symTP_2\rbrace}{\member_\mathsf{drb}}
	\, , \;
	\request{\lbrace\symTP_3,\symTP_4,\symTP_5\rbrace}{\member_\mathsf{kegg}}
	\;
}$
 & 2 \\ \hline
 \end{tabular}
 \caption{Source assignments for FedBench query LS6
	by different federation engines.}
 \label{tb: source-selection-assignment}
 \end{table*}

\subsection{Cost and Minimality}

While exhaustive source assignments are correct, for many cases there are other source assignments that are also correct but have a smaller number of
	$\xxxAlgebraOperatorName{req}$ operators.
Smaller numbers are desirable from a performance perspective because each such
	operator
represents the intention to access a given federation member regarding a particular subquery. Hence, we may use this number as a simple
cost function to compare source assignments and, ultimately, to compare different query decomposition \& source selection approaches.

\begin{mydefinition} \label{def:CostOfSourceAssignments}
The \definedTerm{source access cost}~(\definedTerm{sa-cost}) of a source assignment~$a$, denoted by $\sacost{a}$, is the number of subexpressions of the form $\request{\ifaceReqLangExp}{\member}$ that are contained
	(recursively)
within~$a$.
\end{mydefinition}

	Notice
that, for the exhaustive source assignment~$a$ for a BGP~$\symBGP$ over a triple pattern accessible federation~$\Fed$%
	:
$\sacost{a} = \left|\symBGP\right| \cdot \left|\Fed\right|$. For other source assignments, the sa-cost may be smaller.

\begin{myexample} \label{ex:CostOfSourceAssignments}
The exhaustive source assignment for our example BGP~$\symBGP_\mathsf{ex}$~(cf.\ Example~\ref{ex:QuerySemantics}) over federation~$\Fed_\mathsf{ex}$~(cf.\ Example~\ref{ex:Federation})
	would have
an sa-cost of~6. In contrast, the sa-cost of the source assignments $a_\mathsf{ex}$ and $a'_\mathsf{ex}$ in Example~\ref{ex:SourceAssignment} is 4 and 3, respectively. 
\end{myexample}

Based on our notion of sa-cost, we can also introduce a notion of minimality for source assignments.

\begin{mydefinition} \label{def:MinimalityOfSourceAssignments}
Let $\symBGP$ be a BGP and $\Fed$ be a federation that is triple pattern accessible. A source assignment~$a$ that is correct for~$\symBGP$ over~$\Fed$ is \definedTerm{minimal} for~$\symBGP$ over~$\Fed$ if there
	is no other
source assignment~$a'$ such that $a'$ is
correct for $\symBGP$ over~$\Fed$ and $\sacost{a} > \sacost{a'}$.
\end{mydefinition}

\begin{myexample} \label{ex:MinimalityOfSourceAssignments}
By
	\removable{comparing}
the sa-costs of
$a_\mathsf{ex}$ and $a'_\mathsf{ex}$ in Example~\ref{ex:CostOfSourceAssignments}, we see that $a_\mathsf{ex}$ is \emph{not} minimal for $\symBGP_\mathsf{ex}$ over~$\Fed_\mathsf{ex}$.
Moreover, it is not difficult to check that $a'_\mathsf{ex}$ is minimal for $\symBGP_\mathsf{ex}$ over~$\Fed_\mathsf{ex}$.
\end{myexample}

	Finding minimal source assignments is a problem that
resembles the
	\removable{typical}
optimization problem that existing query decomposition~\&~source selection approaches aim to solve
	for
homogeneous federations~(e.g.,
	\removable{\cite{abdelaziz2017lusail,acosta2011anapsid,charalambidis2015semagrow,saleem2018costfed,schwarte2011fedx,DBLP:journals/tlsdkcs/VidalCAMP16}}%
).
Before we study the complexity of this problem for our more general case of heterogeneous federations, we demonstrate the suitability of
	source assignments
as a formal foundation to capture the output of these existing
	approaches, which also allows us to
show limitations of these approaches.

\subsection{Application to Existing Approaches}
\label{sec:SourceSelection:ExistingApproaches}

	This section illustrates
how our notion of source assignments can be applied to describe the output of several existing source selection
	approaches. To this end,
we consider the query LS6 from the FedBench benchmark~\cite{schmidt2011fedbench}, which consists of a
	BGP with five triple patterns, $\symBGP_\mathsf{LS6} = \{\symTP_1, \symTP_2, \symTP_3, \symTP_4, \symTP_5\}$~(cf.\ Listing~\ref{tb:source-assignment-examples}).


\lstset{language=SQL,morekeywords={PREFIX,java,url}}
\begin{lstlisting}[captionpos=b,
                   label=tb:source-assignment-examples,
%                       caption=Query LS6 of the FedBench benchmark (prefix declarations omitted),
                       caption=FedBench query LS6 (prefix declarations omitted),
                   basicstyle=\ttfamily\small,
                   frame=single]
SELECT ?drug ?title WHERE {
  ?drug db:drugCategory dbc:micronutrient .  # tp1
  ?drug db:casRegistryNumber ?id .           # tp2
  ?keggDrug rdf:type kegg:Drug .             # tp3
  ?keggDrug bio2rdf:xRef ?id .               # tp4
  ?keggDrug purl:title ?title }              # tp5
\end{lstlisting}

The FedBench benchmark specifies a federation consisting of several members that all provide a SPARQL endpoint interface, where
	only some
of these members are potentially relevant for query~LS6. Hence, in terms of our data model, we have a federation~$\Fed_\mathsf{FedBench}$ that, among others, contains the following~%
	\removable{four}
members:
\begin{align*}
	\member_\mathsf{drb} &= (\mathit{DrugBank}, \iface_\mathsf{sparql}) \in \Fed_\mathsf{FedBench},   \\
	\member_\mathsf{kegg} &= (\mathit{KEGG}, \iface_\mathsf{sparql}) \in \Fed_\mathsf{FedBench},       \\
	\member_\mathsf{dbp} &= (\mathit{DBPedia}, \iface_\mathsf{sparql}) \in \Fed_\mathsf{FedBench},      \\
	\member_\mathsf{chebi} &= (\mathit{ChEBI}, \iface_\mathsf{sparql}) \in \Fed_\mathsf{FedBench}.    
\end{align*}

FedBench query LS6 is interesting for our purpose because different authors have used it as an example to describe their query decomposition \& source selection approaches~\cite{charalambidis2015semagrow,saleem2018costfed,schwarte2011fedx}. Hence, the output of these approaches for this query is well documented; yet, the authors present these outputs only
	informally \removable{within a textual description or in figures.}
Given our notion of source assignments, we now can provide a precise formal description.

Hence, based on the informal descriptions provided by the authors, we have
	reconstructed
the respective source assignments that the corresponding SPARQL federation engines would produce for BGP~$\symBGP_\mathsf{LS6}$ over \removable{the federation}~$\Fed_\mathsf{FedBench}$. Table~\ref{tb: source-selection-assignment} presents these source assignments.
We emphasize that all
	these source assignments
are correct for $\symBGP_\mathsf{LS6}$ over $\Fed_\mathsf{FedBench}$.
Yet, they are syntactically different and
have different sa-costs, as also detailed in
	the table.
The source assignment with the lowest sa-cost is found
	by 
	CostFed,
whereas the source assignment of 
SemaGrow has the highest sa-cost%
	\removable{~(in this particular~case)}%
.

\subsection{Expressive Power of Source Assignments} \label{sec:SourceSelection:ExpressivePower}

While the source assignments of the different approaches in the previous section differ from one another~(cf.\ Table~\ref{tb: source-selection-assignment}), we observe that they are all of the same \removable{general} form. A related concept that resembles this specific form of source assignments is Vidal et al.%
	's notion of a ``\emph{SPARQL query decomposition}''~\cite{DBLP:journals/tlsdkcs/VidalCAMP16}.
%
	These observations raise the following questions: how does Vidal at el.'s notion compare to our notion of source assignments and, ultimately, \emph{what is the expressive power of the form of source assignments that it resembles}? \par To address these questions we define Vidal et al.'s notion in terms of the source assignments fragment of FedQPL. 
%
To this end, we first notice that 
	a more restricted version of the grammar is sufficient. 

\begin{mydefinition} \label{def:SourceAssignment:JoinOverUnionClass}
The \definedTerm{joins-over-unions class} of source assignments, denoted by $S_{\Join_{(\cup)}}$,
consists of every source assignment~$a$ that can be constructed from the following grammar,
where
$\xxxAlgebraOperatorName{mj}$,
$\xxxAlgebraOperatorName{mu}$, and
$\xxxAlgebraOperatorName{req}$
are terminal symbols,
$\ifaceReqLangExp$ is a triple pattern or a BGP,
$\member$ is a federation member,
	$\mathit{\Phi}_u$ is a nonempty set of source assignments that can be formed by the construction~$a_u$, and $\mathit{\Phi}_b$ is a nonempty set of source assignments that can be formed by the construction~$a_b$.
\begin{center}%
\begin{tabular}{lcl}
$a$ &
\, $::=$ \quad \,&
$     a_u
\quad | \quad
      \xxxAlgebraOperatorName{mj} \, \mathit{\Phi}_u$
\\[1mm]
$a_u$ &
\, $::=$ \quad \,&
$     a_b
\,\quad | \quad
      \xxxAlgebraOperatorName{mu} \, \mathit{\Phi}_b$
\\[1mm]
$a_b$ &
\, $::=$ \quad \,&
$\request{\ifaceReqLangExp}{\member}$
\end{tabular}
\end{center}
\end{mydefinition}

In addition to restricting the grammar, we need to introduce further syntactic restrictions to accurately capture Vidal et al.'s concept of a SPARQL query decomposition. Namely, in terms of our language, this concept is limited to source assignments in $S_{\Join_{(\cup)}}$ for which it holds that, within any subexpression of the form
$$\munionB{
	\request{\ifaceReqLangExp_1}{\member_1}
	, ...\, ,
	\request{\ifaceReqLangExp_n}{\member_n}
},$$
all $\ifaceReqLangExp_1, ...\, , \ifaceReqLangExp_n$ are the same triple pattern or
BGP; additionally, for source assignments of the form $\mjoin{ a_1, ...\, , a_n }$, there cannot be any triple pattern that occurs in more than one of the subexpressions%
.

	For the following formal definition of
these additional restrictions we introduce the recursive function $\mathsf{subexprs}$ that maps every source assignment~$a$ to a~set of all (sub)expressions contained in~$a$,
\begin{equation*}
\mathsf{subexprs}(a) \!\definedAs \begin{cases}
	\lbrace a \rbrace \hspace{26mm}\text{if $a$ is of the form } \request{\ifaceReqLangExp}{\member} , \\
	\lbrace a \rbrace \cup\, \bigcup_{1 \leq i \leq n}\!\mathsf{subexprs}(a_i)
		\quad \text{if $a$ is of the form } \\[-1.5mm] \hspace{38.5mm} \munion{ a_1, ...\, , a_n } , \\
	\lbrace a \rbrace \cup\, \bigcup_{1 \leq i \leq n}\!\mathsf{subexprs}(a_i)
		\quad \text{if $a$ is of the form } \\[-1.5mm] \hspace{38.5mm} \mjoin{ a_1, ...\, , a_n } .
\end{cases}
\end{equation*}
%

Now we are ready to define the \emph{restricted} joins-over-unions class
	that captures Vidal et al.'s concept of a ``SPARQL query decomposition'' accurately.

\begin{mydefinition} \label{def:SourceAssignment:RestrictedJoinOverUnionClass}
The \definedTerm{restricted joins-over-unions class} of source assignments, denoted by $S^*_{\Join_{(\cup)}}$, consists of every source assignment~$a$ that is in $S_{\Join_{(\cup)}}$ and that has
the following two properties.
\begin{enumerate}
	\itemsep2mm 

	\item  \label{defitem:SourceAssignment:RestrictedJoinOverUnionClass:1}
	for every expression
	of the form $\munionB{
		\request{\ifaceReqLangExp_1}{\member_1}
		, ...\, ,
		\request{\ifaceReqLangExp_n}{\member_n}
	}$ that is in $\mathsf{subexprs}(a)$, it holds that
	    $\ifaceReqLangExp_i = \ifaceReqLangExp_j$ for all $i,j \in \lbrace 1, ...\,, n \rbrace$;




	\item  \label{defitem:SourceAssignment:RestrictedJoinOverUnionClass:2}
	if $a$ is of the form $\mjoin{ a_1, ...\, , a_n }$, then
	$\mathsf{tps}(a_i) \cap \mathsf{tps}(a_j) = \emptyset$ for all~$i,j \in \lbrace 1, ...\,, n \rbrace$, where $\mathsf{tps}(a_k)$ denotes to the set of all the triple patterns mentioned in \removable{subexpression}~$a_k$ for all $k \in \lbrace 1, ...\,, n \rbrace$, i.e.,
	\begin{align*}
		\quad\quad
		\mathsf{tps}(a_k) \definedAs \,
		& \lbrace \symTP \in \ifaceReqLangExp \mid a' \in \mathsf{subexprs}(a_k) \text{ such that } a' \text{ is of}
		\\[-2mm]
		& \hspace{13mm} \text{the form $\request{\ifaceReqLangExp}{\member}$ where $\ifaceReqLangExp$ is a BGP } \rbrace
		\;
		\\[-1.5mm]
		& \cup \, \lbrace \ifaceReqLangExp \mid a' \in \mathsf{subexprs}(a_k) \text{ such that } a' \text{ is of the}
		\\[-2mm]
		& \hspace{9.5mm} \text{form $\request{\ifaceReqLangExp}{\member}$ where $\ifaceReqLangExp$ is a triple pattern } \rbrace
		.
	\end{align*}
\end{enumerate}
\end{mydefinition}

Notice that all exhaustive source assignments~(cf.\ Definition~\ref{def:ExhaustiveSourceAssignment}) are in the class~$S^*_{\Join_{(\cup)}}$, and so are the source assignments of Table~\ref{tb: source-selection-assignment}. More generally, to the best of our knowledge, this class encompasses all types of source assignments that
	the query decomposition~\&~source selection approaches proposed in the literature can produce.
%
A natural question at this point
	is: \emph{does their restriction to consider only source assignments in~$S^*_{\Join_{(\cup)}}$ present an actual limitation in the sense that these approaches
		are inherently unable to find a minimal source assignment in specific~cases}?

The following result shows that this is indeed the case.

\begin{myproposition} \label{prop:LimitationOfRestrictedJoinOverUnionClass}
There exists a BGP~$\symBGP$ and a triple pattern accessible federation~$\Fed$ such that
	all source assignments that are both correct and minimal for $\symBGP$ over~$\Fed$ are not in $S^*_{\Join_{(\cup)}}$.
\end{myproposition}

We prove Proposition~\ref{prop:LimitationOfRestrictedJoinOverUnionClass} by showing the following claim, which is even stronger than the claim in Proposition~\ref{prop:LimitationOfRestrictedJoinOverUnionClass}.

\begin{mylemma} \label{lemma:LimitationOfRestrictedJoinOverUnionClass}
There exists a BGP~$\symBGP$, a triple pattern accessible federation~$\Fed$, and a source assignment~$a$ that is correct for $\symBGP$ over~$\Fed$, such that
	$\sacost{a} < \sacost{a'}$ for every source assignment~$a'$ that is in $S^*_{\Join_{(\cup)}}$ and that is also correct for $\symBGP$ over~$\Fed$.
\end{mylemma}

\begin{proof}[Sketch]
Lemma~\ref{lemma:LimitationOfRestrictedJoinOverUnionClass} can be shown based on our example source assignment~$a'_\mathsf{ex}$~(%
Example~\ref{ex:SourceAssignment}) which is not in $S^*_{\Join_{(\cup)}}$ and has an sa-cost that is smaller than the sa-cost of any relevant
	\removable{source}
assignment in~$S^*_{\Join_{(\cup)}}$.
(For a detailed discussion of this and the other proofs in this section, which we have only sketched due to space limitations, refer to the
	\PaperVersion{appendix in the extended version of this paper~\cite{ExtendedVersion}.)}%
	\ExtendedVersion{appendix on page~\pageref{appendix}.)}%
\end{proof}

As a final remark, we emphasize that the same limitation exists even if we consider the
	less restricted
class $S_{\Join_{(\cup)}}$ (instead of $S^*_{\Join_{(\cup)}}$).

\begin{mylemma} \label{lemma:LimitationOfJoinOverUnionClass}
There exists a BGP~$\symBGP$, a triple pattern accessible federation~$\Fed$, and a source assignment~$a$ that is correct for $\symBGP$ over~$\Fed$, such that
	$\sacost{a} < \sacost{a'}$ for every source assignment~$a'$ that is in $S_{\Join_{(\cup)}}$ and that is also correct for $\symBGP$ over~$\Fed$.
\end{mylemma}

\begin{proof}[Sketch]
The proof of Lemma~\ref{lemma:LimitationOfRestrictedJoinOverUnionClass} also proves Lemma~\ref{lemma:LimitationOfJoinOverUnionClass} because
	\removable{the source assignment} $a'_\mathsf{ex}$ \removable{is also not in} $S_{\Join_{(\cup)}}$.
\end{proof}

\subsection{Complexity of Source Selection} \label{sec:SourceSelection:Complexity}

	For a version of the source selection problem that is defined based on their notion of a ``SPARQL query decomposition,'' Vidal et al.\ show that this problem
is NP-hard~\cite{DBLP:journals/tlsdkcs/VidalCAMP16}. Since we know from the previous section that this notion
	does not provide the full expressive power of our notion of source assignments
and, furthermore,
	Vidal et al.'s
work focuses only on homogeneous federations (of SPARQL endpoints), it is interesting to study the complexity of source selection for our more general case. 
To this end, we formulate the following decision problem.

\begin{mydefinition} \label{def:SourceSelection:DecisionProblem}
Given a BGP~$\symBGP$, a triple pattern accessible federation~$\Fed$, and a positive integer~$c$, the \definedTerm{source selection problem} is to decide whether there exists a source assignment~$a$ such that $a$ is correct for $\symBGP$ over~$\Fed$ and $\sacost{a} \leq c$.
\end{mydefinition}

Unfortunately, it can be shown that this problem is
	also
NP-hard.

\begin{mytheorem} \label{thm:SourceSelectionComplexity:LowerBound}
The source selection problem is NP-hard.
\end{mytheorem}

\begin{proof}[Sketch]
We show the NP-hardness by a reduction from the node cover problem~(also called vertex cover problem), which is known to be NP-hard~\cite{DBLP:conf/coco/Karp72}. The detailed proof
	\PaperVersion{can be found in~\cite{ExtendedVersion}.}%
	\ExtendedVersion{is in the appendix~(page~\pageref{Proof:thm:SourceSelectionComplexity:LowerBound}).}%
\end{proof}

The following theorem also gives an upper bound for the complexity of the source selection problem in Definition~\ref{def:SourceSelection:DecisionProblem} (note that Vidal et al.\ do not show such a result for their version of the source selection problem~\cite{DBLP:journals/tlsdkcs/VidalCAMP16}).

\begin{mytheorem} \label{thm:SourceSelectionComplexity:UpperBound}
The source selection problem is in $\Sigma_2^\mathrm{P}$.
\end{mytheorem}

\begin{proof}[Sketch]
For this proof we assume a nondeterministic Turing machine that guesses a source assignment~$a$ and, then, checks that $a$ is correct for $\symBGP$ over~$\Fed$ and that $\sacost{a} \leq c$. While $\sacost{a} \leq c$ can be checked \removable{in polynomial time} by
scanning $a$, for the correctness check, the machine uses an NP oracle
	\PaperVersion{as detailed in~\cite{ExtendedVersion}.}%
	\ExtendedVersion{(see the appendix).}%
\end{proof}

%% file: equivalences.tex
\section{Equivalences} \label{sec:Equivalences}

While initial logical plans resulting from the query decomposition~\& source selection step can produce the correct query results, they may not be efficient in many cases. A query optimizer can convert such plans into more efficient
	ones by systematically replacing
subexpressions by other subexpressions that are semantically equivalent%
	\removable{~(i.e., that are guaranteed to produce the same result)}%
.
	This section provides a solid \removable{formal} foundation for such optimizations by showing
	\removable{a comprehensive set of} such equivalences~for~FedQPL.

Before we begin, we need to define the notion of \emph{semantic equivalence} for FedQPL expressions which is federation dependent.

\begin{mydefinition} \label{def:SemanticEquivalence}
Let $\Fed$ be a federation%
	. Two FedQPL expressions $\varphi$ and $\varphi'$ that are valid for $\Fed$
are \definedTerm{semantically equivalent} for $\Fed$, denoted by $\varphi \myequiv{\Fed} \varphi'$\!, if it holds that $\sols{\varphi} = \sols{\varphi'}$.
\end{mydefinition}

Now, the first equivalences cover expressions
	that focus on a single federation member \removable{with an interface}
that supports \emph{triple patterns requests}%
	\removable{~(recall from Section~\ref{sec:DataModel} that, among others, this includes the SPARQL endpoint interface, the TPF interface, and the brTPF interface)}%
.
	These \removable{equivalences}
follow from Definition~\ref{def:FedQPL:Semantics}.

\begin{myproposition} \label{prop:Equivalences:TP}
Let $\member = (\symRDFgraph,\iface)$ be a
member in a federation~$\Fed$ such that
	$\iface = \ifaceTuple$ is some interface that
supports triple pattern requests;
let $\symTP$ be a triple pattern, and $\varphi$ and $\varphi'$ be FedQPL expressions that are valid for $\Fed$.
It holds that:
\begin{enumerate}
	\item $\joinB{ \request{\symTP}{\member} \,}{\, \varphi} \,\myequiv{\Fed}\, \tpAdd{\varphi}{\symTP}{\member\!}$;

	\item $\joinB{ \request{\symTP}{\member} \,}{\, \join{\varphi}{\varphi'} } \,\myequiv{\Fed}\, \joinB{ \tpAdd{\varphi}{\symTP}{\member\!} \,}{\, \varphi'}$.
\end{enumerate}
\end{myproposition}

\begin{proof}[Sketch]
Since all equivalences in this section can be shown in a similar manner by applying Definition~\ref{def:FedQPL:Semantics}, we illustrate the proof only for Equivalence~(1) and leave the rest as an exercise for the reader.
To prove Equivalence~(1), we have to show that $\sols{ \join{ \request{\symTP}{\member}}{\varphi} } = \sols{ \tpAdd{\varphi}{\symTP}{\member\!} }$, which we do based on Definition~\ref{def:FedQPL:Semantics} as follows~(where we can assume that $\symTP \in \ifaceReqLang{}$):
\begin{align*}
\sols{ \join{ \request{\symTP}{\member}}{\varphi} }
& = \sols{ \request{\symTP}{\member} } \Join \sols{ \varphi }
& \textsf{\scriptsize by Def.~\ref{def:FedQPL:Semantics}, Case (4)}
\\
& = \ifaceFct(\ifaceReqLangExp,\symRDFgraph) \Join \sols{ \varphi }
& \textsf{\scriptsize by Def.~\ref{def:FedQPL:Semantics}, Case (1)}
\\
& = \sols{ \varphi } \Join \ifaceFct(\ifaceReqLangExp,\symRDFgraph)
& \textsf{\scriptsize by the commuta-}
\\[-2mm]
&& \textsf{\scriptsize tivity of $\Join$~\cite{DBLP:journals/tods/PerezAG09}}
\\
& = \sols{ \tpAdd{\varphi}{\symTP}{\member\!} }
& \textsf{\scriptsize by Def.~\ref{def:FedQPL:Semantics}, Case (2)}
.
\\[-10mm]
\end{align*}%
\end{proof}

\begin{myexample} \label{ex:Equivalences:TP}
By applying Equivalence~(1)~(cf.\ Proposition~\ref{prop:Equivalences:TP}),
	we may rewrite the FedQPL expression in Example~\ref{ex:tpAdd} into the following expression which is semantically equivalent for~$\Fed_\mathsf{ex}$.
\begin{equation*}
\join{ \request{\symTP_2}{\member_2} }{ \request{\symTP_1}{\member_1} }
\end{equation*} 
\end{myexample}

The following equivalences focus on expressions with a federation member
	that
supports \emph{BGP requests}. These equivalences follow from
	the definition of BGPs\removable{~(cf.\ Section~\ref{sec:Preliminaries})}, in combination with
Definition~\ref{def:FedQPL:Semantics}.
\removable{Notice that the first two
	\removable{of these~}%
equivalences
	are the BGP-specific counterparts of the equivalences
in Proposition~\ref{prop:Equivalences:TP}.}

\begin{myproposition} \label{prop:Equivalences:BGP}
Let $\member = (\symRDFgraph,\iface)$ be a member in a federation $\Fed$ such that
	$\iface$ is some interface that
supports BGP requests;
let $\symBGP$, $\symBGP_1$, $\symBGP_2$, and $\symBGP'$ be BGPs such that $\symBGP'\! = \symBGP_1 \cup \symBGP_2$, and $\varphi$ and $\varphi'$ be FedQPL expressions that are valid for $\Fed$.
It holds that:
\begin{enumerate}
\setcounter{enumi}{2}
	\item $\joinB{\request{\symBGP}{\member} \,}{\, \varphi} \,\myequiv{\Fed}\, \bgpAdd{\varphi}{\symBGP}{\member\!}$;

	\item $\joinB{ \request{\symBGP}{\member} \,}{\, \join{\varphi}{\varphi'} } \,\myequiv{\Fed}\, \joinB{ \bgpAdd{\varphi}{\symBGP}{\member\!} \,}{\, \varphi'}$;

	\item $\joinB{ \request{\symBGP_1}{\member} \,}{\, \request{\symBGP_2}{\member\!} } \,\myequiv{\Fed}\, \request{\symBGP'}{\member\!}$; 

	\item $\bgpAddB{ \request{\symBGP_2}{\member\!} }{\symBGP_1}{\member\!} \,\myequiv{\Fed}\, \request{\symBGP'}{\member\!}$; 

	\item $\bgpAddB{ \bgpAdd{\varphi}{\symBGP_2}{\member\!} }{\symBGP_1}{\member\!} \,\myequiv{\Fed}\, \bgpAdd{\varphi}{\symBGP'}{\member\!}$. 
\end{enumerate}
\end{myproposition}

The following equivalences focus on
	\removable{expressions with a} federation member
whose interface supports \emph{both, triple patterns requests and BGP requests}. These equivalences follow from
	the definition of BGPs, in combination with
Definition~\ref{def:FedQPL:Semantics}.

\begin{myproposition} \label{prop:Equivalences:TP&BGP}
Let $\member = (\symRDFgraph,\iface)$ be a member in a federation~$\Fed$ such that
	$\iface$ is some interface that
supports triple pattern requests
	as well as 
BGP requests;
let $\symTP$ be a triple pattern, $\symBGP = \lbrace \symTP_1, \symTP_2, ...\,, \symTP_n \rbrace$ be a BGP, and $\varphi$ be a FedQPL expression that is valid for $\Fed$.
It holds that:
\begin{enumerate}
\setcounter{enumi}{7}
	\item $\request{\symTP}{\member}  \,\myequiv{\Fed}\,  \request{\symBGP'}{\member}$, where $\symBGP' = \lbrace \symTP \rbrace$;

	\item $\request{\symBGP}{\member\!}  \,\myequiv{\Fed}\,  \joinB{ \joinB{\ldots \joinB{ \request{\symTP_1}{\member\!} }{ \request{\symTP_2}{\member\!} }  }{ \;\ldots \;} }{ \request{\symTP_n}{\member\!} }$;

	\item $\request{\symBGP}{\member\!}  \,\myequiv{\Fed}\,  \tpAdd{ \ldots (\tpAdd{ \request{\symTP_1}{\member\!} }{\symTP_{2\!}}{\member\!}) \ldots  }{\symTP_n}{\member\!}$;

	\item $\bgpAdd{\varphi}{\symBGP}{\member}  \,\myequiv{\Fed}\,  \tpAdd{ \ldots (\tpAdd{ \tpAdd{\varphi}{\symTP_1\!}{\member\!} }{\symTP_{2\!}}{\member\!}) \ldots  }{\symTP_n\!}{\member\!}$;

	\item $\bgpAdd{ \request{\symTP}{\member} }{\symBGP}{\member\!}  \,\myequiv{\Fed}\,  \request{ \symBGP' }{\member}$, where $\symBGP'\! = \symBGP \cup \lbrace\symTP\rbrace$;

	\item $\tpAdd{ \request{\symBGP}{\member} }{\symTP}{\member\!}  \,\myequiv{\Fed}\,  \request{ \symBGP' }{\member}$, where $\symBGP'\! = \symBGP \cup \lbrace\symTP\rbrace$;

	\item $\tpAddB{ \bgpAdd{\varphi}{\symBGP}{\member\!} }{\symTP}{\member\!}  \myequiv{\Fed}  \bgpAdd{\varphi}{\symBGP'}{\member\!}$, where $\symBGP'\! = \symBGP \cup \lbrace\symTP\rbrace$.
\end{enumerate}
\end{myproposition}

The following equivalences focus on
	\removable{expressions with}
a federation member
	\removable{that provides}
the \emph{brTPF interface}~(cf.\ Example~\ref{ex:Interface:brTPF}).
	\removable{Consequently,} these
equivalences follow from the definition of that interface, in combination with Definition~\ref{def:FedQPL:Semantics}.

\begin{myproposition}
Let $\member = (\symRDFgraph,\iface_\mathsf{brTPF})$ be a
member in a federation~$\Fed$ such that $\iface_\mathsf{brTPF}$ is the brTPF interface;
let $\symTP$ be a triple pattern, $\Omega$~be a set of solution mappings, and $\varphi$ be a FedQPL expression that is valid for $\Fed$.
It holds that:
\begin{enumerate}
	\itemsep1mm 
\setcounter{enumi}{14}
	\item $\joinB{\request{\symTP}{\member} \,}{\, \varphi}  \,\myequiv{\Fed}\,  \request{(\symTP,\Omega)}{\member\!}$, where $\Omega = \sols{\varphi}$;

	\item $\tpAdd{\varphi}{\symTP}{\member\!}  \,\myequiv{\Fed}\,  \request{(\symTP,\Omega)}{\member\!}$, where $\Omega = \sols{\varphi}$.
\end{enumerate} 
\end{myproposition}

The following equivalences focus on
	\removable{expressions with}
federation members
	\removable{that provide}
the \emph{SPARQL endpoint interface}~(cf.\ Example~\ref{ex:Interface:SPARQL}).
	They
follow from the definition of SPARQL graph patterns~\cite{DBLP:journals/tods/PerezAG09,Harris13:SPARQL1_1Language}, in combination with Definition~\ref{def:FedQPL:Semantics}.
\todo{TODO: more equivalences are possible for SPARQL endpoints; e.g., simulating something like brTPF requests by using FILTER or by using the VALUES clause (see Carlos' paper for these things)}

\begin{myproposition}
Let $\member = (\symRDFgraph,\iface_\mathsf{sparql})$ be a
member in a federation~$\Fed$ where $\iface_\mathsf{sparql}$ is the SPARQL endpoint interface;
let $\symTP$ be a triple pattern, $\symBGP$ be a BGP, and $\symPattern$, $\symPattern_1$ and $\symPattern_2$ be graph patterns.
It holds~that:
\begin{enumerate}
	\itemsep1mm 
\setcounter{enumi}{16}
	\item $\request{\symPattern_1}{\member\!}  \,\myequiv{\Fed}\,  \request{\symPattern_2}{\member\!}$ \, if $\symPattern_1$ and $\symPattern_2$ are semantically equivalent~\cite{DBLP:journals/tods/PerezAG09};

	\item $\unionB{ \request{\symPattern_1}{\member} \,}{\, \request{\symPattern_2}{\member\!} }  \,\myequiv{\Fed}\,  \request{ (\symPattern_1 \,\mathrm{UNION}\, \symPattern_2) }{\member\!}$;

	\item $\joinB{ \request{\symPattern_1}{\member} \,}{\, \request{\symPattern_2}{\member\!} }  \,\myequiv{\Fed}\,  \request{ (\symPattern_1 \,\mathrm{AND}\, \symPattern_2) }{\member\!}$;

	\item \removable{$\tpAdd{ \request{\symPattern}{\member} }{\symTP}{\member}  \,\myequiv{\Fed}\,  \request{ (\symPattern \,\mathrm{AND}\, \symTP) }{\member\!}$;}

	\item \removable{$\bgpAdd{ \request{\symPattern}{\member} }{\symBGP}{\member}  \,\myequiv{\Fed}\,  \request{ (\symPattern \,\mathrm{AND}\, \symBGP) }{\member\!}$;}
\end{enumerate} 
\end{myproposition}

The following equivalences cover expressions
	with
\emph{multiple federation members}.
	These equivalences
follow from Definition~\ref{def:FedQPL:Semantics}.

\begin{myproposition}
Let $\member_1 \!=\! (\symRDFgraph,\iface_1)$, $\member_2 \!=\! (\symRDFgraph,\iface_2)$, $\member_3 \!=\! (\symRDFgraph,\iface_3)$, and $\member_4 \!=\! (\symRDFgraph,\iface_4)$ be 
members in a federation~$\Fed$~(not necessarily different ones) such that interfaces $\iface_1$ and $\iface_2$ support triple pattern requests and interfaces $\iface_3$ and $\iface_4$ support BGP requests.
Let $\symTP$, $\symTP_1$ and $\symTP_2$ be triple patterns,
$\symBGP$, $\symBGP_1$ and $\symBGP_2$ be BGPs,
and $\varphi$ be a FedQPL expression that is valid for $\Fed$.
It holds that:
\begin{enumerate}
	\itemsep1mm 
\setcounter{enumi}{21}
	\item $\tpAddB{ \tpAdd{\varphi}{\symTP_2}{\member_2\!} }{\symTP_1}{\member_1\!}  \,\myequiv{\Fed}\,  \tpAddB{ \tpAdd{\varphi}{\symTP_1}{\member_1\!} }{\symTP_2}{\member_2\!}$;

	\item $\tpAddB{ \bgpAdd{\varphi}{\symBGP}{\member_3\!} }{\symTP}{\member_1\!}  \,\myequiv{\Fed}\,  \bgpAddB{ \tpAdd{\varphi}{\symTP}{\member_1\!} }{\symBGP}{\member_3\!}$;

	\item $\bgpAddB{ \bgpAdd{\varphi}{\symBGP_2}{\member_4\!} }{\symBGP_1}{\member_3\!}  \,\myequiv{\Fed}\,  \bgpAddB{ \bgpAdd{\varphi}{\symBGP_1}{\member_3\!} }{\symBGP_2}{\member_4\!}$.
\end{enumerate} 
\end{myproposition}

The following equivalences are independent of interface types and focus on the relationships between the two multiway operators~($\xxxAlgebraOperatorName{mj}$ and $\xxxAlgebraOperatorName{mu}$) and their respective binary counterparts.

\begin{myproposition}
Let $\Fed$ be a federation,
and let $\mathit{\Phi}$, $\mathit{\Phi}_\mathsf{J}$, and $\mathit{\Phi}_\mathsf{U}$ be sets of FedQPL expressions that are valid for $\Fed$ such that
$|\mathit{\Phi}|>1$,
there exists $\join{\varphi_1}{\varphi_2} \in \mathit{\Phi}_\mathsf{J}$%
	\ and
$\union{\varphi_1}{\varphi_2} \in \mathit{\Phi}_\mathsf{U}$.
It holds that:
\begin{enumerate}
	\itemsep1mm 
\setcounter{enumi}{24}
	\item $\xxxAlgebraOperatorName{mj} \,\mathit{\Phi}  \,\myequiv{\Fed}\,  \join{ \xxxAlgebraOperatorName{mj} \, \mathit{\Phi}'\! }{\, \varphi }$, where $\varphi \in \mathit{\Phi}$ and $\mathit{\Phi}' = \mathit{\Phi} - \lbrace \varphi \rbrace$;

	\item $\xxxAlgebraOperatorName{mu} \,\mathit{\Phi}  \,\myequiv{\Fed}\,  \union{ \xxxAlgebraOperatorName{mu} \, \mathit{\Phi}'\! }{\, \varphi }$, where $\varphi \in \mathit{\Phi}$ and $\mathit{\Phi}' = \mathit{\Phi} - \lbrace \varphi \rbrace$;

	\item $\xxxAlgebraOperatorName{mj} \,\mathit{\Phi}_\mathsf{J}  \,\myequiv{\Fed}\,  \xxxAlgebraOperatorName{mj} \, \mathit{\Phi}'$\!, where $\mathit{\Phi}'\! = ( \mathit{\Phi}_\mathsf{J} - \lbrace \join{\varphi_1}{\varphi_2} \rbrace ) \cup \lbrace \varphi_1,\varphi_2 \rbrace$;  

	\item $\xxxAlgebraOperatorName{mu} \,\mathit{\Phi}_\mathsf{U}  \myequiv{\Fed}  \xxxAlgebraOperatorName{mu} \, \mathit{\Phi}'$\!, where $\mathit{\Phi}'\! = ( \mathit{\Phi}_\mathsf{U} - \lbrace \union{\varphi_1}{\varphi_2} \rbrace ) \cup \lbrace \varphi_1,\varphi_2 \rbrace$.  
\end{enumerate}
\end{myproposition}

The following equivalences are also independent of interface types and focus only on the two multiway operators.

\begin{myproposition}
Let $\Fed$ be a federation;
let $\varphi$ be a FedQPL expression that is valid for $\Fed$,
and let $\mathit{\Phi}_\mathsf{J}$ and $\mathit{\Phi}_\mathsf{U}$ be sets of FedQPL expressions that are valid for $\Fed$ such that there exists  $\xxxAlgebraOperatorName{mj}\,\mathit{\Phi}_\mathsf{J}'\! \in \mathit{\Phi}_\mathsf{J}$ and $\xxxAlgebraOperatorName{mu}\,\mathit{\Phi}_\mathsf{U}'\! \in \mathit{\Phi}_\mathsf{U}$.
It holds that:
\begin{enumerate}
	\itemsep1mm 
\setcounter{enumi}{28}
	\item $\xxxAlgebraOperatorName{mj} \,\mathit{\Phi}_\mathsf{J} \,\myequiv{\Fed}\,  \xxxAlgebraOperatorName{mj} \, \mathit{\Phi}_\mathsf{J}''$, where $\mathit{\Phi}_\mathsf{J}'' = ( \mathit{\Phi}_\mathsf{J} - \lbrace \xxxAlgebraOperatorName{mj}\,\mathit{\Phi}_\mathsf{J}' \rbrace ) \cup \mathit{\Phi}_\mathsf{J}'$;

	\item $\xxxAlgebraOperatorName{mu} \,\mathit{\Phi}_\mathsf{U} \,\myequiv{\Fed}\,  \xxxAlgebraOperatorName{mu} \, \mathit{\Phi}_\mathsf{U}''$, where $\mathit{\Phi}_\mathsf{U}'' = ( \mathit{\Phi}_\mathsf{U} - \lbrace \xxxAlgebraOperatorName{mu}\,\mathit{\Phi}_\mathsf{U}' \rbrace ) \cup \mathit{\Phi}_\mathsf{U}'$;

	\item $\munion{ \varphi } \,\myequiv{\Fed}\,  \varphi$;
	\item $\mjoin{ \varphi } \,\myequiv{\Fed}\,  \varphi$.
\end{enumerate}
\end{myproposition}

Lastly, the following equivalences are also independent of interface types; they
follow
from Definition~\ref{def:FedQPL:Semantics} and the corresponding equivalences for the SPARQL algebra~\cite{DBLP:journals/tods/PerezAG09,DBLP:conf/icdt/Schmidt0L10}.

\begin{myproposition} \label{prop:Equivalences:join&union}
Let $\Fed$ be a federation and let $\varphi_1$, $\varphi_2$, and $\varphi_3$ be FedQPL expressions that are valid for $\Fed$.
It holds that:
\begin{enumerate}
	\itemsep1mm 
\setcounter{enumi}{32}
	\item $\join{\varphi_1}{\varphi_2}  \,\myequiv{\Fed}\,  \join{\varphi_2}{\varphi_1}$;
	\item $\union{\varphi_1}{\varphi_2}  \,\myequiv{\Fed}\,  \union{\varphi_2}{\varphi_1}$;
	\item $\union{\varphi_1}{\varphi_1}  \,\myequiv{\Fed}\,  \varphi_1$;
	\item $\joinB{\varphi_1}{\join{\varphi_2}{\varphi_3}}  \,\myequiv{\Fed}\,  \joinB{\join{\varphi_1}{\varphi_2}}{\varphi_3}$;
	\item $\unionB{\varphi_1}{\union{\varphi_2}{\varphi_3}}  \,\myequiv{\Fed}\,  \unionB{\union{\varphi_1}{\varphi_2}}{\varphi_3}$;
	\item $\joinB{\varphi_1}{\union{\varphi_2}{\varphi_3}}  \,\myequiv{\Fed}\,  \unionB{\join{\varphi_1}{\varphi_2}}{\join{\varphi_1}{\varphi_3}}$.
\end{enumerate}
\end{myproposition}

\begin{myexample} \label{ex:Equivalences:join&union}
By applying Equivalence~(33)~(cf.\ Proposition~\ref{prop:Equivalences:join&union}), we may swap the two subexpressions in the FedQPL expression in Example~\ref{ex:Equivalences:TP}, which results in the following expression.
\begin{equation*}
\join{ \request{\symTP_1}{\member_1} }{ \request{\symTP_2}{\member_2} }
\end{equation*}
We may now rewrite this expression into the following one by applying Equivalence~(1)~(cf.\ Proposition~\ref{prop:Equivalences:TP}) again.
\begin{equation*}
\tpAddB{ \request{\symTP_2}{\member_2\!} }{\symTP_1}{\member_1\!}
\end{equation*}
Observe that the latter expression is, thus, semantically equivalent (for $\Fed_\mathsf{ex}$) not only to the previous expression, but also to the expressions in Examples~\ref{ex:Equivalences:TP} and~\ref{ex:tpAdd}, respectively. 
\end{myexample}

\begin{myexample} \label{ex:tpAdd-2}
The previous example
highlights that an alternative to the plan captured by the FedQPL expression
$\tpAddB{ \request{\symTP_1}{\member_1\!} }{\symTP_2}{\member_2\!}$
 in Example~\ref{ex:tpAdd} would be---among others---the plan represented by the expression $\tpAddB{ \request{\symTP_2}{\member_2\!} }{\symTP_1}{\member_1\!}$.
Considering that the interface of federation member~$\member_1$ is
	a brTPF interface
whereas $\member_2$ has a TPF interface~(cf.\ Example~\ref{ex:Federation}), we may prefer the latter plan because it 
	provides for a greater number of algorithms to choose from during
physical query optimization~(e.g., implementing $\xxxAlgebraOperatorName{tpAdd}$ over
	brTPF may be done by using a bind join algorithm~\cite{hartig2016brtpf}, which is impossible with TPF).
On the other hand, a query optimizer may estimate a significantly smaller result size for $\request{\symTP_1}{\member_1\!}$ than for $\request{\symTP_2}{\member_2\!}$, which could justify choosing the plan of Example~\ref{ex:tpAdd}.
\end{myexample}

%% file: concluding.tex
\section{Concluding Remarks} \label{sec:Concluding}

In this paper, in addition to proving several important results regarding the query plan language that we propose%
, we have demonstrated initial ideas for applying this language.

For instance, we have shown that our source assignments can
represent the output of existing query decomposition \& source selection approaches~(cf.\ Section~\ref{sec:SourceSelection:ExistingApproaches}), and that these approaches are inherently
	limited%
~(cf.\ Proposition~\ref{prop:LimitationOfRestrictedJoinOverUnionClass}). Hence, there is still an opportunity for future work on better source selection approaches even in the context of homogeneous federations of SPARQL endpoints.

Moreover, we have not only demonstrated that the language can be used to represent logical query plans~(cf.\ Examples~\ref{ex:request}--\ref{ex:union} and~\ref{ex:FedQPL:Correctness}), but also that it is suitable both as a basis for logical query optimization~(cf.\ Examples~\ref{ex:Equivalences:TP}--\ref{ex:Equivalences:join&union}) and as a starting point for physical query optimization~(cf.\ Example~\ref{ex:tpAdd-2}).
	As future work regarding FedQPL, we are planning to extend the language with additional operators that can be used to represent query plans for more expressive fragments of SPARQL, and we aim to also provide a multiset semantics.

	However,
the ultimate next step
is to develop effective optimization approaches for queries over federations with heterogeneous interfaces. FedQPL provides a formal foundation for~such~work.

%% file: appendix.tex
\newpage
\appendix

\begin{center}
	\huge
\bfseries
Appendix
\end{center}

\section*{Proof of Lemma~\ref{lemma:LimitationOfRestrictedJoinOverUnionClass}}

\label{appendix}

\textbf{Lemma~\ref{lemma:LimitationOfRestrictedJoinOverUnionClass}.} \,
There exists a BGP~$\symBGP$, a triple pattern accessible federation~$\Fed$, and a source assignment~$a$ that is correct for $\symBGP$ over~$\Fed$, such that $\sacost{a} < \sacost{a'}$ for every source assignment~$a'$ that is in $S^*_{\Join_{(\cup)}}$ and that is also correct for $\symBGP$ over~$\Fed$.

\begin{proof}
Recall our example source assignment~$a'_\mathsf{ex}$ for
BGP~$\symBGP_\mathsf{ex}$ over
	\removable{our example}
federation~$\Fed_\mathsf{ex}$~(cf.\ Example~\ref{ex:SourceAssignment}). $\Fed_\mathsf{ex}$ is triple pattern accessible~(cf.\ Example~\ref{ex:Federation}) and $a'_\mathsf{ex}$ is correct for $\symBGP_\mathsf{ex}$ over $\Fed_\mathsf{ex}$~(cf.\ Example~\ref{ex:SourceAssignment}). Note also that $a'_\mathsf{ex}$ is clearly not in $S^*_{\Join_{(\cup)}}$%
. Now, we may enumerate every source assignment~$a'$ in $S^*_{\Join_{(\cup)}}$ that is valid for $\symBGP_\mathsf{ex}$ over $\Fed_\mathsf{ex}$ and that has an sa-cost smaller or equivalent to 3, i.e., $\sacost{a'} \leq \sacost{ a'_\mathsf{ex} }$. For each of these source assignments we will find that it is not correct. 
\end{proof}

\section*{Proof of Lemma~\ref{lemma:LimitationOfJoinOverUnionClass}}

\textbf{Lemma~\ref{lemma:LimitationOfJoinOverUnionClass}.} \,
There exists a BGP~$\symBGP$, a triple pattern accessible federation~$\Fed$, and a source assignment~$a$ that is correct for $\symBGP$ over~$\Fed$, such that $\sacost{a} < \sacost{a'}$ for every source assignment~$a'$ that is in $S_{\Join_{(\cup)}}$ and that is also correct for $\symBGP$ over~$\Fed$.

\begin{proof}
	We can use the same proof as for Lemma~\ref{lemma:LimitationOfRestrictedJoinOverUnionClass}
because the source assignment~$a'_\mathsf{ex}$ as used in that proof is also not in $S_{\Join_{(\cup)}}$.
\end{proof}

\section*{Proof of Theorem~\ref{thm:SourceSelectionComplexity:LowerBound}}

\label{Proof:thm:SourceSelectionComplexity:LowerBound}

\textbf{Theorem~\ref{thm:SourceSelectionComplexity:LowerBound}.}
The source selection problem is NP-hard.

\begin{proof}
We show the NP-hardness by a reduction from the node cover problem~(also called vertex cover problem), which is known to be NP-hard~\cite{DBLP:conf/coco/Karp72}.

The node cover problem is defined as follows: Given
	a positive integer~$k$ and an undirected graph~$G=(V,E)$,
decide whether there exists a set~$V'\! \subseteq V$ such $\left| V' \right| \leq k$ and every edge in $E$ is incident on some vertex in $V'\!$%
	\removable{~(i.e., for every undirected edge~$\{u,v\} \in E$ it holds that $u \in V'$ or $v \in V'$)}%
.

For our reduction we introduce the following function~$f$ that maps every instance of the node cover problem~(i.e., every
	$G=(V,E)$ and~$k$)
to an instance of the source selection problem~(i.e., a BGP~$\symBGP$, a triple pattern accessible federation~$\Fed$, and a positive integer~$c$).
Given $G=(V,E)$ and~$k$, function~$f$ maps the \removable{undirected} graph~$G$ to a federation~$\Fed$ that consists of one member per vertex in $V$. Hence, we define~$\Fed$ such that $\left|\Fed\right| = \left|V\right|$ and there exists a bijection $f_\mathsf{memb}\!: V \rightarrow \Fed$. \removable{Then,} for every \removable{vertex}~$v \in V$, the corresponding federation member $f_\mathsf{memb}(v)=(\symRDFgraph_v,\iface_v)$ is defined as follows: \removable{The interface~}$\iface_v$ is the TPF interface~(cf.\ Example~\ref{ex:Interface:TPF}),
	\removable{or any other interface that supports triple pattern requests,}
and for the RDF graph~$\symRDFgraph_v$ we have that
\begin{align*}
\symRDFgraph_v = \Big\lbrace \bigl( \mathsf{uri}(e), \mathtt{rdf\!:\!type}, \mathtt{ex\!:\!Edge} \bigr) \,\Big\vert\,
& \text{$e$ is an edge in $E$}
\\[-3mm]
& \text{that is incident on $v$ }
 \Big\rbrace ,
\end{align*}
where
	$\mathsf{uri}\!: E \rightarrow \symAllURIs$ is a bijection that maps every edge in $E$ to a distinct URI.

In addition to mapping the input graph~$G$ to the federation~$\Fed$, function~$f$ maps the integer~$k$ directly to $c$%
	\removable{~(i.e., $c=k$)},
and the BGP~$\symBGP$
	\removable{returned by $f$}
is the same for every instance of the node cover problem;
	\removable{this BGP consists of a single triple pattern:} $\symBGP = \lbrace \symTP \rbrace$
where $\symTP = \bigl( ?x, \mathtt{rdf\!:\!type}, \mathtt{ex\!:\!Edge} \bigr)$.
It is not difficult to see that our mapping function~$f$ can be computed in polynomial time.

Then, the reduction is based on the following claim: For any possible input $G=(V,E)$ and~$k$, and the corresponding output $(\symBGP, \Fed, c) = f(G,k)$,
there exists a set~$V'\! \subseteq V$ such $\left| V' \right| \leq k$ and every edge in $E$ is incident on some vertex in $V'$ if and only if there exists a source assignment~$a$ such that $a$ is correct for $\symBGP$ over~$\Fed$ and $\sacost{a} \leq c$.
%
%
This claim is easily verified based on
two~observations:
\begin{enumerate}
\item For every edge $e=\lbrace u,v \rbrace$ in $E$, there are exactly two federation members
	that contain
the RDF triple $$\bigl( \mathsf{uri}(e), \mathtt{rdf\!:\!type}, \mathtt{ex\!:\!Edge} \bigr),$$ namely, the members
	\removable{created for the vertexes $v$ and $u$, that is,}
$f_\mathsf{memb}(v)$ and~$f_\mathsf{memb}(u)$.
\item For at least one of these two members, say~$\member$,
	a subexpression of the form $\request{\ifaceReqLangExp}{\member}$ must be part of every source assignment that is correct for $\symBGP$ over~$\Fed$.
\end{enumerate}

Hence, the node cover problem can be reduced to our source selection problem, and since it is NP-hard%
	~\cite{DBLP:conf/coco/Karp72}%
, the source selection problem must be NP-hard%
	~as~well%
.
\end{proof}

\section*{Proof of Theorem~\ref{thm:SourceSelectionComplexity:UpperBound}}

Before we present the proof of Theorem~\ref{thm:SourceSelectionComplexity:UpperBound}, we show two auxiliary results which we shall then use to prove the theorem.

\begin{mylemma} \label{lemma:FedEvaluationOfBgpIsPTIME}
Given a BGP~$\symBGP$, a triple pattern accessible federation~$\Fed$, and a solution mapping~$\mu$, the problem
	to decide
whether $\mu \in \eval{\symBGP}{\Fed}$ can be solved in
	time $O( |\Fed|^2 + |\symBGP| \cdot |\Fed| )$.
\end{mylemma}

\begin{proof}
To check whether $\mu \in \eval{\symBGP}{\Fed}$ we actually have to check whether $\mu \in \eval{\symBGP}{\symRDFgraph_\mathrm{union}}$ where $\symRDFgraph_\mathrm{union} \!= \bigcup_{(\symRDFgraph,\iface) \in \Fed} \symRDFgraph$~(cf.\ Definition~\ref{def:QuerySemantics}). To this end, we first materialize the RDF graph~$\symRDFgraph_\mathrm{union}$, which is possible in
	time $O( |\Fed|^2 )$ (note that the complexity is quadratic because of the need to eliminate duplicates when materializing $\symRDFgraph_\mathrm{union}$). Thereafter,
checking whether $\mu \in \eval{\symBGP}{\symRDFgraph_\mathrm{union}}$ is know as the evaluation problem of SPARQL, which has been shown to be solvable in time $O( |\symBGP| \cdot |\symRDFgraph_\mathrm{union}| )$ for~BGPs~\cite{DBLP:journals/tods/PerezAG09}. Since $|\Fed|=|\symRDFgraph_\mathrm{union}|+k$ for some constant $k$, the algorithm has an overall time complexity of $O( |\Fed|^2 + |\symBGP| \cdot |\Fed| )$.
\end{proof}

\begin{mylemma} \label{lemma:MuInSolsIsPTIME}
Given a BGP~$\symBGP$, a triple pattern accessible federation~$\Fed$, a source assignment~$a$ that is valid for $\symBGP$ over~$\Fed$, and a solution mapping~$\mu$, the problem of deciding whether $\mu \in \mathsf{sols}(a)$ can be solved in
	time $O( |a| \cdot |\Fed| + |a| \cdot |\mu| )$.
\end{mylemma}

\begin{proof}
We proof the lemma by induction over the structure of~$a$.

\smallskip\noindent
\textit{Base case:} If $a$ is of the form $\request{\ifaceReqLangExp}{\member}$, we can
	use Algorithm~\ref{algo:MuInSolsIsPTIME:BaseCase} to check whether $\mu \in \mathsf{sols}(a)$.
By Definition~\ref{def:SourceAssignment}, we know that $\ifaceReqLangExp$ is a triple pattern or a BGP. W.l.o.g., we assume it is a BGP~$\symBGP'$\!.
In the algorithm, we write $\fctVars{a}$ to denote the set of variables mentioned in the BGPs within the source assignment~$a$. Formally, this set is defined recursively as follows: If $a$ is of the form $\request{\ifaceReqLangExp}{\member}$ where $\ifaceReqLangExp$ is a triple pattern or a BGP, then $\fctVars{a} = \fctVars{\ifaceReqLangExp}$. If $a$ is of the form $\munion{ a_1, ...\, , a_n }$ or of the form $\mjoin{ a_1, ...\, , a_n }$, then $\fctVars{a} = \bigcup_{i \in \lbrace 1,...,n\rbrace}\fctVars{a_i}$.

The first step of the algorithm~(lines~\ref{step:MuInSolsIsPTIME:BaseCase:IfBegin}--\ref{step:MuInSolsIsPTIME:BaseCase:IfEnd}) checks that $\mu$ is actually defined for the variables in $a$, which can be done in time $O( |a|\cdot|\mu| )$. The next step is to replace all variables in $\symBGP'$ according to $\mu$~(cf.\ line~\ref{step:MuInSolsIsPTIME:BaseCase:Substitute}), which results in a set~$\symRDFgraph'$ of RDF triples. The time complexity of this step is again $O( |a|\cdot|\mu| )$ (we assume here that $|a|=|\symBGP'|+k$ for some constant $k$). Finally, the algorithm checks that every triple in $\symRDFgraph'$ exists in the data~$\symRDFgraph$ of federation member~$\member$~(lines~\ref{step:MuInSolsIsPTIME:BaseCase:ForBegin}--\ref{step:MuInSolsIsPTIME:BaseCase:ForEnd}). The actual check for each triple~(i.e., line~\ref{step:MuInSolsIsPTIME:BaseCase:IfInsideFor}) can be done in time $O(|\Fed|)$ (we use that $|\symRDFgraph|<|\Fed|$) and, thus, the time complexity of the whole for loop~(lines~\ref{step:MuInSolsIsPTIME:BaseCase:ForBegin}--\ref{step:MuInSolsIsPTIME:BaseCase:ForEnd}) is $O( |a|\cdot|\Fed| )$. Therefore, the time complexity of the whole algorithm is $O( |a| \cdot |\mu| + |a| \cdot |\Fed| )$.

\begin{algorithm}[t!]
\caption{Check whether $\mu \in \mathsf{sols}(a)$ for $a = \request{\symBGP'}{\member}$ with $\member = (\symRDFgraph,\iface)$}\label{algo:MuInSolsIsPTIME:BaseCase}
{\small
\begin{algorithmic}[1]
\IF {$\fctDom{\mu} \neq \fctVars{a}$} \label{step:MuInSolsIsPTIME:BaseCase:IfBegin}
	\RETURN \FALSE
	\hspace{9mm}\COMMENT{$\mu$ has to be defined for the variables in $a$}
\ENDIF \label{step:MuInSolsIsPTIME:BaseCase:IfEnd}

\STATE let $\symRDFgraph' := \mu[\symBGP']$ \label{step:MuInSolsIsPTIME:BaseCase:Substitute}
\hspace{10mm}\COMMENT{substitute all variables in $\symBGP'$ according to $\mu$}

\FORALL {$t \in \symRDFgraph'$} \label{step:MuInSolsIsPTIME:BaseCase:ForBegin}
	\IF {$t \notin \symRDFgraph$}  \label{step:MuInSolsIsPTIME:BaseCase:IfInsideFor}
		\RETURN \FALSE
			\hspace{6mm}\COMMENT{search the data of $\member$ for every triple in $\symRDFgraph'$}
	\ENDIF
\ENDFOR \label{step:MuInSolsIsPTIME:BaseCase:ForEnd}
\RETURN \TRUE
\end{algorithmic}}
\end{algorithm}

\smallskip\noindent
\textit{Induction step:} In the induction step we consider the remaining two cases of $a$.

Case 1) If $a$ is of the form $\munion{ a_1, ...\, , a_n }$, we can check whether $\mu \in \mathsf{sols}(a)$ by using Algorithm~\ref{algo:MuInSolsIsPTIME:IndStepUnion}. The algorithm tries to find a sub-expression~$a_i$ inside $a$ such that we have $\mu \in \mathsf{sols}(a_i)$. By the induction hypothesis, the corresponding check in line~\ref{step:MuInSolsIsPTIME:IndStepUnion:IfInsideFor} can be done in time $O( |a_i| \cdot |\mu| + |a_i| \cdot |\Fed| )$ for every $i \in \lbrace 1,...,n \rbrace$. Consequently, for the whole loop~(lines~\ref{step:MuInSolsIsPTIME:IndStepUnion:ForBegin}--\ref{step:MuInSolsIsPTIME:IndStepUnion:ForEnd}), and thus the whole algorithm, we have a time complexity of $O\bigl( (|a_1|+...+|a_n|) \cdot |\mu| + (|a_1|+...+|a_n|) \cdot |\Fed| \bigr)$, which is $O( |a| \cdot |\mu| + |a| \cdot |\Fed| )$.

\begin{algorithm}[t!]
\caption{Check whether $\mu \in \mathsf{sols}(a)$ for $a = \munion{ a_1, ...\, , a_n }$}\label{algo:MuInSolsIsPTIME:IndStepUnion}
{\small
\begin{algorithmic}[1]
\FORALL {$i \in \lbrace 1,...\,,n \rbrace$} \label{step:MuInSolsIsPTIME:IndStepUnion:ForBegin}
	\IF {$\mu \in \mathsf{sols}(a_i)$}  \label{step:MuInSolsIsPTIME:IndStepUnion:IfInsideFor}
		\RETURN \TRUE
	\ENDIF
\ENDFOR \label{step:MuInSolsIsPTIME:IndStepUnion:ForEnd}
\RETURN \FALSE
\end{algorithmic}}
\end{algorithm}

Case 2) If $a$ is of the form $\mjoin{ a_1, ...\, , a_n }$, we can check whether $\mu \in \mathsf{sols}(a)$ by using Algorithm~\ref{algo:MuInSolsIsPTIME:IndStepJoin}. The first step~(lines~\ref{step:MuInSolsIsPTIME:IndStepJoin:IfBegin}--\ref{step:MuInSolsIsPTIME:IndStepJoin:IfEnd}) is to check that $\mu$ is actually defined for the variables in $a$, which can be done in time $O( |a|\cdot|\mu| )$. Thereafter, the algorithm iterates over the subexpressions $a_1$ to $a_n$.

For each such
	subexpression
$a_i$, the algorithm first takes the restriction of $\mu$ to the variables in $\fctVars{a_i}$, denoted by $\mu_i$~(line~\ref{step:MuInSolsIsPTIME:IndStepJoin:Projection}); i.e., $\mu_i$ is a solution mapping such that $\fctDom{\mu_i} = \fctDom{\mu} \cap \fctVars{a_i}$ and $\mu_i(?v)=\mu(?v)$ for every variable $?v \in \fctDom{\mu} \cap \fctVars{a_i}$. For every $a_i$, this steps can be done in time $O( |a_i|\cdot|\mu| )$. Next, the algorithm checks whether $\mu_i \notin \mathsf{sols}(a_i)$~(line~\ref{step:MuInSolsIsPTIME:IndStepJoin:IfInsideFor}), in which case $\mu$ cannot be in $\mathsf{sols}(a)$. By the induction hypothesis, for every $i \in \lbrace 1,...\,,n \rbrace$, this check can be done in time $O( |a_i| \cdot |\mu_i| + |a_i| \cdot |\Fed| )$, which we may generalize to $O( |a_i| \cdot |\mu| + |a_i| \cdot |\Fed| )$ because $|\mu|=|\mu_i|+k_i$ for some constant~$k_i$.

Then, the time complexity of the whole for loop~(lines~\ref{step:MuInSolsIsPTIME:IndStepJoin:ForBegin}--\ref{step:MuInSolsIsPTIME:IndStepJoin:ForEnd})
is $O\bigl( (|a_1|+...+|a_n|) \cdot |\mu| + (|a_1|+...+|a_n|) \cdot |\Fed| \bigr)$, which is $O( |a| \cdot |\mu| + |a| \cdot |\Fed| )$. When combined with the complexity of the first step~(%
	lines~\ref{step:MuInSolsIsPTIME:IndStepJoin:IfBegin}--\ref{step:MuInSolsIsPTIME:IndStepJoin:IfEnd},
see above), the time complexity of Algorithm~\ref{algo:MuInSolsIsPTIME:IndStepJoin} is also $O( |a| \cdot |\mu| + |a| \cdot |\Fed| )$.
\end{proof}

\begin{algorithm}[t!]
\caption{Check whether $\mu \in \mathsf{sols}(a)$ for $a = \mjoin{ a_1, ...\, , a_n }$}\label{algo:MuInSolsIsPTIME:IndStepJoin}
{\small
\begin{algorithmic}[1]
\IF {$\fctDom{\mu} \neq \fctVars{a}$} \label{step:MuInSolsIsPTIME:IndStepJoin:IfBegin}
	\RETURN \FALSE
	\hspace{9mm}\COMMENT{$\mu$ has to be defined for the variables in $a$}
\ENDIF \label{step:MuInSolsIsPTIME:IndStepJoin:IfEnd}

\FORALL {$i \in \lbrace 1,...\,,n \rbrace$} \label{step:MuInSolsIsPTIME:IndStepJoin:ForBegin}
	\STATE let $\mu_i$ be the restriction of $\mu$ to the variables in $\fctVars{a_i}$ \label{step:MuInSolsIsPTIME:IndStepJoin:Projection}
	\IF {$\mu_i \notin \mathsf{sols}(a_i)$}  \label{step:MuInSolsIsPTIME:IndStepJoin:IfInsideFor}
		\RETURN \FALSE
	\ENDIF
\ENDFOR \label{step:MuInSolsIsPTIME:IndStepJoin:ForEnd}
\RETURN \TRUE
\end{algorithmic}}
\end{algorithm}

Now we are ready to prove the theorem.

\medskip\noindent
\textbf{Theorem~\ref{thm:SourceSelectionComplexity:UpperBound}.}
The source selection problem is in $\Sigma_2^\mathrm{P}$.

\begin{proof}
	We
assume a nondeterministic Turing machine~(NTM) that is equipped with
	\removable{the following}
oracle. For every BGP~$\symBGP$, every triple pattern accessible federation~$\Fed$, and every source assignment~$a$, the oracle
	returns \texttt{true} if and only if
$a$ is \emph{not} correct for $\symBGP$ over~$\Fed$.
Then, the NTM decides the source selection problem for any given input $\symBGP$, $\Fed$, and $c$ as follows:
First, the NTM guesses a source assignment~$a$ such that the size of $a$ is polynomial in the size of $\symBGP$ and $\Fed$~(by Proposition~\ref{prop:CorrectnessOfExhaustiveSourceAssignments} we know that such polynomial-sized source assignments exist and are correct for $\symBGP$ over~$\Fed$). Then, the NTM has to check
	that $a$ is correct for $\symBGP$ over~$\Fed$ and that $\sacost{a} \leq c$. The latter property, i.e., $\sacost{a} \leq c$,
can be checked in polynomial time by scanning $a$ and counting all subexpressions of the form $\request{\ifaceReqLangExp}{\member}$. To check the correctness of~$a$ the NTM uses its oracle (and inverts the response of the oracle; i.e., $a$ is correct for $\symBGP$ over~$\Fed$ if and only if the oracle returns \texttt{\small false}).

Now,
	\removable{to show that the source selection problem is in $\Sigma_2^\mathrm{P}$}
it remains to show that checking whether $a$ is not correct for $\symBGP$ over~$\Fed$~(i.e., the decision problem solved by the oracle) is in NP.
To this end, we use the following nondeterministic program: First, we test whether $a$ is valid for $\symBGP$ over~$\Fed$, which is a precondition for the correctness~(cf.\ Definition~\ref{def:FedQPL:Correctness}) and can be checked in polynomial time if we assume that the encoding of $\Fed$ on the tape of a Turing Machine includes an indication of the type of interface that each member in $\Fed$ has. Next, we guess a solution mapping~$\mu$ where the intuition is that $\mu \in \eval{\symBGP}{\Fed}$ but $\mu \notin \mathsf{sols}(a)$, which shows that $a$ is not correct for $\symBGP$ over~$\Fed$. Hence, the program has to check these two properties of $\mu$, which can be done in polynomial time as we have shown in Lemmas~\ref{lemma:FedEvaluationOfBgpIsPTIME} and~\ref{lemma:MuInSolsIsPTIME}.
\end{proof}